\documentclass[letterpaper,11pt]{article}
\usepackage{amsmath, amssymb, amsthm}
\usepackage{color}
\usepackage{algorithm}
\usepackage[noend]{algpseudocode}
\usepackage{tikz}
\usepackage{enumerate}
\usepackage{graphicx}
\usepackage{caption}
\usepackage[colorlinks=true,linkcolor=blue,citecolor=magenta
]{hyperref}
\usepackage{subcaption}
\usepackage{amsfonts}
\usepackage{comment}
\usepackage{algorithmicx}
\usepackage{thmtools} 
\usepackage{mathtools}
\usepackage{thm-restate}
\usepackage[listings]{tcolorbox}
\tcbuselibrary{listings,theorems}

\usepackage{tikz} 
\usetikzlibrary{positioning,chains,fit,shapes,calc,arrows}
\usepackage{tikz-cd}

\DeclarePairedDelimiter\floor{\lfloor}{\rfloor}

\usepackage{ctable}					
\newtheorem{claim}{Claim}[section]

\newtheorem{observation}{Observation}[section]
\usepackage{mathtools}
\usepackage[normalem]{ulem}

\usepackage{footnote}
\makesavenoteenv{tabular}
\makesavenoteenv{table}

\usepackage[margin=1in]{geometry}

\newtheorem{example}{Example}
\newtheorem{property}{Property}
\newtheorem{theorem}{Theorem}[section]
\newtheorem{corollary}{Corollary}[section]
\newtheorem{lemma}[theorem]{Lemma}

\newtheorem{definition}[theorem]{Definition}

\makeatletter
\def\GrabProofArgument[#1]{ #1: \egroup\ignorespaces}
\def\proof{\noindent\textbf\bgroup Proof%
	\@ifnextchar[{\GrabProofArgument}{. \egroup\ignorespaces}}
\def\endproof{\hspace*{\fill}$\Box$\medskip}
\makeatother

\definecolor{navyblue}{rgb}{0.0, 0.0, 0.5}

\newcommand{\agent}{i}

\newcommand{\agents}{N}
\newcommand{\goods}{M}
\newcommand{\good}{g}
\newcommand{\pool}{P}

\newcommand{\MMS}{\mathsf{MMS}}

\newcommand{\util}{v}

\newcommand{\EFX}{\mathsf{EFX}}
\newcommand{\EFO}{\mathsf{EF1}}
\newcommand{\EFkX}{\mathsf{EFkX}}
\newcommand{\EFXX}{\mathsf{EF2X}}
\newcommand{\EFXp}{\mathsf{EFX^+}}

\newcommand{\lexlarger}{\underset{\scriptscriptstyle\textsf{lex}}\succ}
\newcommand{\lexlargereq}{\underset{\scriptscriptstyle\textsf{lex}}\succeq}

\newcommand{\valu}{v}
\newcommand{\alloc}{X}
\newcommand{\allocs}{X}

\newcommand{\etal}{\textit{et al.}}
\newcommand{\champset}[2]{[{#1} \mid {#2}]}

\newcommand{\suchthat}{s.t. }

\newcounter{proccnt}

\usepackage{textcomp}

\newcommand{\konote}[1]{}

\DeclareMathOperator*{\argmin}{arg\,min}
\DeclareMathOperator*{\argmax}{arg\,max}
\begin{document}
\title{An $\EFXX$ Allocation Protocol for Restricted Additive Valuations}
\author{
		Hannaneh Akrami \thanks{MPII, SIC, Graduate School of Computer Science, Saarbrücken} \and Rojin Rezvan \thanks{University of Texas at Austin} \and Masoud Seddighin \thanks{School of Computer Science, Institute for Research in Fundamental Sciences (IPM), Tehran, Iran} 
}

\maketitle

\makeatletter
\def\GrabProofArgument[#1]{ of \cref{#1}. \egroup\ignorespaces}
\def\proof{\noindent\textbf\bgroup Proof%
	\@ifnextchar[{\GrabProofArgument}{. \egroup\ignorespaces}}
\def\endproof{\hspace*{\fill}\qed \medskip}
\makeatother                

\begin{abstract}
	We study the problem of fairly allocating a set of $m$ indivisible goods to a set of $n$ agents.  Envy-freeness up to any good ($\EFX$) criteria --- which 
requires that no agent prefers the bundle of another agent after  removal of any single good --- is known to be a remarkable analogous of envy-freeness when the resource is a set of indivisible goods \cite{caragiannis2019unreasonable}.
In this paper, we investigate $\EFX$ notion for the restricted additive valuations, that is, every good has some non-negative
value, and every agent is interested in only some of the goods.

We introduce a natural relaxation of $\EFX$ called $\EFkX$ which requires that no agent envies another agent after  removal of any $k$ goods.  
Our main contribution is an algorithm that finds a complete (i.e., no good is discarded) $\EFXX$ allocation for the restricted additive valuations. In our algorithm we devise new concepts, namely ``configuration'' and ``envy-elimination'' that might be of independent interest. 

We also use our new tools to find an $\EFX$ allocation for restricted additive valuations that discards at most $\floor{n/2} -1$ goods. This improves the currently best known result of Chaudhury \textit{et al.}~\cite{charity} for the restricted additive valuations by a factor of $2$. 
\end{abstract}

\section{Introduction}
\label{intro}

Fair allocation deals with the problem of allocating a resource to agents with diverse preferences. Due to the wide range of applications, this problem has received attention in different fields such as economics, mathematics, operations research, politics, and computer science \cite{lipton2004approximately,saha2010new,asadpour2008santa,feige2008allocations,brams1996fair,Foley:first,brams,Bezacova:first,Steinhaus:first}. 

In the early studies, the resource was assumed to be a single heterogeneous divisible cake. This case has been mostly considered by mathematicians and economists under the title of ``Cake Cutting".
 In contrast,  recent studies often focus on more practical cases where the resource is less divisible or indivisible. Such instances arise in many real-world scenarios, e.g., dividing the inherited wealth among heirs, divorce settlements, border disputes, etc \cite{10.1287/opre.2016.1544, 10.5555/3033138, 10.5555/3180776}.
In the indivisible case, the resource is a set $\goods$ of $m$ indivisible goods that must be divided among a set of $n$ agents. Each agent $i$ has a valuation function $\valu_i : 2^\goods \rightarrow \mathbb{R}$ and the goal is to allocate the goods to the agents while satisfying a certain fairness objective. 

Envy-freeness is one of the most well-established fairness notions studied in the literature. An allocation is envy-free if each agent prefers her share over other agents' share.  Formally, let $X = \langle X_1,X_2,\ldots,X_n\rangle$ be an allocation that allocates bundle $X_i$ to agent $i$. We say agent $i$ envies agent $j$, if $
\valu_i(X_i) < \valu_i(X_j).
$ An allocation is envy-free, if no agent envies another agent.

Perhaps one of the reasons that envy-freeness is widely accepted among economists is that, despite strict conditions, there are strong guarantees for this notion in the divisible setting. For example, there always exists an envy-free allocation of the cake such that each agent receives a connected piece \cite{stromquist1980cut}. However, beyond divisibility when dealing with a set of indivisible goods, this notion is too strong to be satisfied. For example, consider an instance with two agents and one good; the  agent that receives no good envies the other agent.
Such barriers have led to natural relaxations of envy-freeness that are more suitable for the case of indivisible goods. Envy-freeness up to one good ($\textsf{EF1}$) and envy-freeness up to any good ($\EFX$) are among the most prominent relaxations of envy-freeness for indivisible goods. The idea behind these two notions is to allow a limited amount of envy among the agents. Formally, given allocation $X = \langle X_1,X_2,\ldots,X_n\rangle$ that allocates set $X_i$ of goods to agent $i$, we say $X$ is
\begin{itemize}
	\item $\textsf{EF1}$ (Budish~\cite{budish2011combinatorial}): if for every agent $i$ and $j$, there exists a good  
	$
	 g \in X_j$ such that $\valu_i(X_i) \geq \valu_i(X_j \setminus\{g\}).
	$
	\item $\EFX$  (Caragiannis \etal~\cite{caragiannis2019unreasonable}): if for every agent $i$ and $j$, and every good 
	$g \in X_j$ we have $\valu_i(X_i) \geq \valu_i(X_j \setminus\{g\}).$\footnote{The original definition of $\EFX$ \cite{caragiannis2019unreasonable}, assumes that the removed good $g$ has an additional property that $v_i(X_j \setminus \{g\}) < v_i(X_j)$. In Section \ref{OPTvsEFX}, we denote this more restricted definition by $\EFXp$.
	However, our existential results in Sections \ref{ef2x} and \ref{efxn2} work for the more general case where $v_i(X_j \setminus \{g\}) \leq v_i(X_j)$.}
\end{itemize} 

By definition, every envy-free allocation is also $\EFX$, and every $\EFX$ allocation is also $\EFO$. So far, we know that when the valuation functions are monotone, $\EFO$ allocations always exist and can be found in polynomial time \cite{lipton2004approximately}. In sharp contrast, it turns out that $\EFX$ notion is much more challenging. For example, the only currently known positive result for  $\EFX$  is that for 3 agents with additive valuations an $\EFX$ allocation always exists \cite{chaudhury2020efx}. For example, the only currently known positive results on existence of $\EFX$ allocations are when the valuations are monotone and identical \cite{plaut2020almost}, $n=2$ \cite{plaut2020almost}, or $n=3$ and agents have additive valuations \cite{chaudhury2020efx}.

Recent findings on $\EFX$ suggest that avoiding a subset of goods may result in strong $\EFX$ guarantees. This subject (also known as $\EFX$ with charity)  is pioneered by the work of Caragiannis \textit{et al.} \cite{caragiannis2019envy} wherein the authors show that 
there exists a procedure that discards a subset of goods and finds an $\EFX$ allocation for the rest of the goods such that the Nash welfare\footnote{The Nash welfare of an allocation is the geometric mean of the values of agents.} of the allocation is at least half of the optimal Nash welfare.
Several follow-up works have reduced the number and the total value of the discarded goods \cite{charity, almost-EFX-4}. 

In this paper, we focus on the $\EFX$ notion and its relaxations (including $\EFX$ with charity) when the valuation functions are restricted additive. Restricted additive setting is an important subclass of additive valuations that has gained popularity in allocation problems during the past decade \cite{bansal2006santa,feige2008allocations,asadpour2008santa,cheng2019restricted,davies2020tale,khot2007approximation, Polacek2012QuasipolynomialLS, annamalai2015combinatorial, inproceedings, cheng2018integrality, 10.1145/1993636.1993718}. In Section \ref{restricted_additive} we discuss the restricted additive valuations. Also, we refer to Section \ref{contributions} for the results and the techniques used in this paper.

\subsection{Restricted Additive Valuations}\label{restricted_additive}
In the restricted additive setting, the assumption is that the valuation functions are additive, and furthermore, each good $\good$ has an inherent value $v(g)$ so that for any agent $\agent$, we have $\util_{\agent}(\good)$$ \in \{0,v(g)\}$. On the practical side, this setting captures many real-life scenarios. For example, when distributing food among people, the nutritious value received by each person only depends on the food. However, due to allergies, ethics, etc. people might have different diets and do not eat certain food. Thus, the nutritious value each person receives by each food $g$ is $0$ if the person refuses to eat it or is $v(g)$ which solely depends on the food.

On the theoretical side, the restricted additive class can be considered as a promising middle ground between the identical setting (all the valuations are similar) and the additive setting. For many allocation objectives, there is a considerable discrepancy between the results pertaining to these two settings. For example:
\begin{enumerate}
	\item Maximin-share fairness, or $\MMS$, can be guaranteed for the case of identical valuations, while it is proved that guaranteeing $\MMS$ for the additive setting is  not always possible \cite{kurokawa2018fair}. The best known approximation guarantee for the additive case is $3/4-\epsilon$~\cite{garg2021improved}. 
	\item Guaranteeing $\textsf{EFX}$ for identical additive valuations  is easy \cite{plaut2020almost}, while whether or not an $\textsf{EFX}$ allocation exists for the additive setting is unknown. The best known approximation guarantee for $\textsf{EFX}$ in the  additive setting is  $0.618$  \cite{amanatidis2020multiple,farhadi2020almost}.
	
	\item The best polynomial time approximation guarantee for the Nash welfare objective  with additive valuations is $1.45$ \cite{barman2018finding}, 
	while there is a greedy algorithm with a factor of $1.061$ \cite{barman2018greedy} and a PTAS \cite{nguyen2014minimizing} for the identical additive setting.
	
	\item The best polynomial time approximation guarantee for max-min fairnees in the additive setting is $O(n^\epsilon)$ \cite{chakrabarty2009allocating, bateni2009maxmin}, while there exists a PTAS for the case that the valuations are identical and additive \cite{woeginger1997polynomial}.
\end{enumerate}

Considering this gap, it is expedient to study restricted sub-classes of additive valuations that lie in between these two settings. What makes the restricted additive setting interesting is that despite the simple and handy appearance, its structure is fundamentally different from the identical setting. Indeed, it is believed that for some fairness objectives, the restricted additive setting is as hard as the general additive case \cite{khot2007approximation}. 
Except for max-min fairness  which has a very interesting and fruitful literature for the restricted additive setting (termed as the Santa Claus problem), very few studies consider this type of valuations for other notions. 
For max-min fairness, the best known approximation factor for the  additive setting is improved in series of works to $O(\sqrt n (\log n)^3)$\cite{asadpour2008santa}, $O(\sqrt{n \log n}/\log\log n)$ \cite{saha2010new}, and $O(n^\epsilon)$ \cite{chakrabarty2009allocating,bateni2009maxmin}. Extensive investigations  on the restricted additive valuations for max-min fairness has also resulted in constant factor approximation algorithms \cite{annamalai2015combinatorial,  asadpour2008santa,davies2020tale}. 
The first breakthrough for this setting is the work of Bansal and Sviridenko \cite{bansal2006santa}, who provide an $O(\log\log n/\log\log\log n)$-approximation solution based on rounding a certain linear program called configuration LP. 
Based on an impressive result of Haxel for hypergraph matchings \cite{haxell1995condition}, Asadpour \etal~\cite{asadpour2008santa}  prove that the integrality gap of Configuration LP is within a  factor $4$ for the restricted additive setting.  This bound is slightly improved to $3.84$ by Jansen and Rohwedder \cite{jansen2018compact} and Cheng and Mao~\cite{cheng2018integrality}. 
Recently, Bamas \etal~\cite{bamas2020submodular} introduced a submodular version of restricted valuations for the Santa Claus problem.

To the best of our knowledge, no previous work considers restricted additive valuations for other objectives such as maximin-share fairness, $\EFX$, and Nash social welfare. Recently, another class of valuations called bi-valued is studied for Nash welfare ‌in which the value of each good to each agent is either $p \in \mathbb{N}$ or $q \in \mathbb{N}$ \cite{akrami2021nash}. 

 In this paper, we initiate the study of $\EFX$ notation, when the valuations are restricted additive.  We refer to Section \ref{contributions} for the results and the techniques used in this paper.

\subsection{Our Contribution}\label{contributions}
We start by introducing  $\EFkX$ notion which is indeed a  relaxation of $\EFX$ that allows an amount of envy up to the value of $k$ least valuable goods of a bundle. 
Our main result is Algorithm \ref{alg2} that finds a complete $\EFXX$ allocation for restricted additive valuations. The algorithm consists of 3 updating rules plus an additional final step. As long as it is possible, we update the allocation using one of the updating rules. When none of the rules is applicable, we perform the additional step to obtain a complete allocation. 
 The rules are based on new concepts, namely configuration and envy-elimination which we describe in the following.

\paragraph{Configuration.} One important point of departure of our method from the existing techniques  is that alongside updating the allocation, we maintain a partitioning of the agents into several groups. This partitioning has the property that the value of the agents in the same group are close to each other.  We use the term configuration to refer to a pair of an allocation and a partition. The updating process at each step takes a configuration as input and updates both the allocation and the partition.  

Well-established concepts such as champion and champion-graph are also  revised in accordance with the definition of configuration.

\paragraph{Envy-elimination.} At the heart of our updating rules we exploit a process called envy-elimination. Envy-elimination is designed to circumvent deadlocks. At several points during the algorithm we might allocate a good to an agent that violates the $\EFX$ property. In such situations, we execute the envy-elimination process. 
This process restores the $\EFX$ property by merely eliminating goods from the bundles of agents. Therefore, at the end of this process, the value of each agent  for her bundle is  at most as much as her value beforehand. 

\paragraph{A new Potential Function.} Note that the fact that the social welfare strictly decreases after the envy-elimination process might question the termination of our algorithm.  This brings us to another challenge: we must show that the algorithm ends after a finite number of updates. 
To prove the termination of the algorithm, we introduce a potential function $\Phi$ which maps
a pair $\sigma = (X,R)$ of a partial allocation $X$  and a partition $R$ of agents to a vector $\Phi(\sigma)$ of rational numbers and show that after each update, $\Phi(\sigma)$ 
increases lexicographically. This indicates that the updating process terminates after a finite number of updates.  Note that for a given configuration $\sigma = (X,R)$, function $\Phi(\sigma)$  relies on both $X$ and $R$.

	 We later turn our attention beyond $\EFXX$ to see whether our new tools can be used to obtain better guarantees for the $\EFX$ notion. Our second result is Algorithm \ref{alg1} that finds a partial $\EFX$ allocation which discards at most $\floor{n/2}-1$ number of goods.  The currently best known  result in this direction is the work of Chaudhury \etal~\cite{charity} which proves the existence of a partial $\EFX$ allocation that discards at most $n-1$ goods. In comparison to them, our algorithm reduces the number of discarded goods for the restricted additive setting by a factor of 2.

  Finally, in Section \ref{OPTvsEFX}, we explore the possibility of achieving more efficient allocations satisfying $\EFkX$. We show that for any integer $k$ and any constant $c$, there are instances for which guaranteeing a $c$-approximation of $\EFkX$ and Pareto efficiency at the same time is not possible. 

\section{Related Work}
Fair allocation of indivisible goods is a central problem in several disciplines including computer science and economics \cite{brams1996fair}, and envy-freeness is a classic fairness notion studied in this context. A formal study of envy-freeness can be traced back to over 70 years ago \cite{Foley:first,Steinhaus:first}.  The notion we study in this paper is $\EFX$ which is a relaxation of envy-freeness for the case that the resource is a set of indivisible goods.

$\EFX$  is  among the most studied fairness notions in recent years and “Arguably, the best fairness analog of envy-freeness for indivisible items” \cite{caragiannis2019unreasonable}. This notion originates in the work of Caragiannis \textit{et al.}~\cite{caragiannis2019unreasonable} wherein the authors provide some initial results on $\EFX$ and its relation to other notions.
However, despite extensive investigations, the existence of a complete $\EFX$ allocation is only proved for very limited cases: when the number of agents is $2$ or $3$ \cite{plaut2020almost,chaudhury2020efx}, and 
when the valuations are either identical  \cite{plaut2020almost}, binary \cite{barman2018greedy}, or bi-valued \cite{amanatidis2021maximum}. 
Given this impenetrability of $\EFX$, a growing strand of research started considering its relaxations. These relaxations can be classified into three categories:
\begin{itemize}
	\item \textbf{Approximately $\EFX$ allocations}:  a natural approach is to find allocations that are approximately $\EFX$. The first result in this direction is a $1/2$-$\EFX$ allocation proposed by Plaut and Roughgarden \cite{plaut2020almost}. The approximation is later improved to $0.618$ by Amanitidis \textit{et al.} \cite{amanatidis2020multiple}. 
	
	\item \textbf{Weakening the fairness requirement}: recall that in an $\EFX$ allocation, any possible envy is removed by eliminating the least valuable good. Recently, Farhadi \textit{et al.} \cite{farhadi2020almost}  suggest a relaxed version of $\EFX$ called $\textsf{EFR}$ in which instead of eliminating the least valuable good to evaluate fairness, eliminates a good uniformly at  random.  They also show that a $0.74$-$\textsf{EFR}$ allocation always exists. Another example is envy-freeness up to one less preferred good $(\textsf{EFL})$ introduced by Barman \textit{et al.} \cite{barman2018groupwise} which limits the value of the eliminated good. The $\EFkX$ notion we introduced in Section \ref{contributions} also falls within this category.
		
	\item \textbf{Discarding a subset of goods}: another approach is to relax the assumption that the final allocation must allocate all the goods. This line is initiated by the work of Caragiannis \etal~\cite{caragiannis2019envy} wherein the authors prove the existence of a partial $\EFX$ allocation whose Nash welfare is at least half of the optimal Nash welfare. Following this work, Chaudhury \textit{et al.} \cite{charity} proved the existence of a partial $\EFX$ allocation that leaves at most $n-1$ goods unallocated.  Berger \textit{et al.} \cite{almost-EFX-4} decreased the number of discarded goods to $n-2$.
	Recently,  Chaudhury \textit{et al.} \cite{rainbow} presented a framework to obtain a partial $(1 - \epsilon)$-$\EFX$ allocation with sub-linear number of unallocated goods for any $\epsilon \in (0, \frac{1}{2}]$. Our second result in Section \ref{efxn2} falls 
	within this category.
\end{itemize}

Finally, we note that there are other fairness and efficiency criteria for the allocation of indivisible goods. The most prominent examples are maximin-share \cite{asadpour2008santa, bansal2006santa, bateni2009maxmin, Bezacova:first, chakrabarty2009allocating, cheng2018integrality, cheng2019restricted, davies2020tale, feige2008allocations, garg2021improved, ghodsi2018fair, khot2007approximation, kurokawa2018fair, Polacek2012QuasipolynomialLS} and its variants such as pairwise maximin-share \cite{amanatidis2020multiple} and group-wise maximin-share \cite{amanatidis2020multiple, barman2018groupwise}, envy-freeness up to one good \cite{barman2018finding, lipton2004approximately}, Nash welfare \cite{barman2018finding, caragiannis2019unreasonable, farhadi2020almost, akrami2021nash, amanatidis2021maximum, barman2018greedy}, and competitive equilibrium from equal incomes \cite{10.1287/opre.2016.1544}. We also refer the reader to \cite{10.5555/3033138,brams1996fair} for an overview on fair division, classic fairness notions, and related results.

\section{Preliminaries}
\label{preliminaries}

We denote the set of agents by $\agents = \{1, 2, \ldots, n\}$ and the set of goods by $\goods$. Each agent $i$ has a valuation function $\valu_i : 2^\goods \rightarrow \mathbb{R}^+$ which represents the value of that agent for each subset of the goods.
An allocation is a specification of how goods in $\goods$ are divided among the agents. We denote an allocation by $\allocs=\langle \alloc_1,\alloc_2,\ldots,\alloc_n \rangle$, where $\alloc_i$ is the bundle allocated to agent $i$. Allocation $X$ is complete if $\bigcup_i \alloc_i = \goods,$ and is partial otherwise. For a partial allocation $X$, we refer to the set of goods that are not allocated to any agent as the pool of unallocated goods and denote it by $\pool_{\allocs}$. When $X$ is clear from the context, we simply use $P$ instead of $\pool_{\allocs}$. 
We say a good $g$ is \textit{wasted}, if it is allocated to an agent that has zero value for $g$. Typically, non-wasteful allocations refer to allocations that admit no wasted good. 

In this paper, we are interested in allocations that satisfy certain fairness properties. In Section \ref{intro}, we defined the $\EFX$ notion.
Here we define  a more general form of $\EFX$, namely, envy-freeness up to any $k$ goods or $\EFkX$. 

\begin{definition}\label{efkxdef}
	An allocation $X = \langle \alloc_1, ... , \alloc_n \rangle$ is $\EFkX$, if for all $i \neq j$ and every collection of $\ell = \min(k, |X_j|)$ distinct goods $g_1,g_2,\ldots,g_\ell \in \alloc_j$  we have 
	$
	v_i(\alloc_i) \geq v_i(\alloc_j \setminus\{g_1,g_2,\ldots,g_\ell\}).
	$
\end{definition}

Roughly speaking, an allocation is  $\EFkX$ if for every  agents $i,j$, agent $i$ does not envy agent $j$ provided that $k$ arbitrary goods are removed from  the bundle of agent $j$.
In particular, an $\EFX$ allocation is $\EFkX$ with $k=1$.  Our main results are related to the cases where $k=1$ and $k=2$.

\paragraph{Restricted Additive Valuations.}
We consider a special case of valuation functions in which each good $\good$ has an inherent value $\valu(\good)$ and the contribution of $g$ to any set of goods for any agent is either $0$ or  $\valu(\good)$. Under this restriction, when the valuations are also additive, we call them restricted additive. 
\begin{definition}\label{rest-add-def}
	A set $ \{\valu_1,\valu_2,\ldots,\valu_n\}$ of valuation functions is restricted additive, if for every $1 \leq i \leq n$, $\valu_i$ is additive, and furthermore, for every good $\good \in \goods$, we have $\valu_i(\good)\footnote{For ease of notation, we use $\util_i(\good)$ instead of $\util_i(\{\good\})$.} \in \{0,v(\good)\}$ where $v(g)>0$.
\end{definition}
For brevity, for a set $S$ of goods we define $v(S)$ as $\sum_{g \in S} v(g)$. One desirable property of restricted additive valuations is stated in Observation \ref{non-waste-good}. 

\begin{observation}\label{non-waste-good}
	Let $X$ be a partial $\EFX$ allocation, let $g$ be an unallocated good, and let $i$ and $j$ be two different agents. If there exists two different goods $g_1, g_2 \in X_j \cup \{g\}$ such that $v_i((X_j \cup \{g\}) \setminus \{g_1, g_2\}) > v_i(X_i)$, then $v_i(g) \neq 0$. 
\end{observation}
\begin{proof}
	Assume $v_i(g) = 0$. Since $g_1 \neq g_2$, we have $g \neq g_1$ or $g \neq g_2$.
	Without loss of generality, assume $g \neq g_1$. We have
	\begin{align*}
		v_i((X_j \cup \{g\}) \setminus \{g_1, g_2\}) &= v_i(X_j \setminus \{g_1, g_2\}) \\ &\leq v_i(X_j \setminus \{g_1\}) \\
		&\leq v_i(X_i), &\hbox{$X$ is $\EFX$}
	\end{align*}
	which is a contradiction. Therefore, $v_i(g) \neq 0$. Since $v_i$ is restricted additive, we have $v_i(g) = v(g)$.
\end{proof}

\subsection{Configurations}
In this work, we represent the status of our algorithm at each step via two ingredients: a partial allocation $\alloc$ and a \emph{partition}\footnote{{A partition of a set is a grouping of its elements into non-empty subsets, in such a way that every element is included in exactly one subset.}} $R$ of the agents. 
For a partition $R$, we denote the $i$'th group of $R$ by $R_i$. Also, we denote by $|R|$ the number of groups in $R$.  

For brevity, we use the term configuration to refer to a pair of an allocation and a partition. Therefore, at each step during our algorithms we take a configuration $(X,R)$ as input and return another configuration $(X',R')$. Our algorithms are developed in such a way that the configurations satisfy three structural properties. The first and the most important property is Property \ref{prop1} which requires that the allocation satisfies the $\EFX$ criteria. 

\begin{property}[EFX]\label{prop1}
	A configuration $\sigma = (X,R)$ is $\EFX$, if allocation $X$ is $\EFX$.
\end{property}

Note that Property \ref{prop1} is independent of the partition.
The second property we consider for configurations is envy-compatibility which we bring in the following. This property incorporates both the allocation and the configuration. 

\begin{property}[Envy-compatibility]\label{def-envy-conf}
	A configuration $(X,R)$ is envy-compatible, if for all $1 \leq \ell \leq |R|$ we have
	$$
	\forall_{i \in R_\ell, \good \in X_i} \qquad \valu_i(X_i \setminus \{\good\}) \leq \valu_{r_\ell}(X_{r_\ell}),
	$$
	where 
	\begin{equation}\label{sl}
	r_\ell = \arg \min_{j \in R_\ell} \valu_j(X_j).
	\end{equation} 
\end{property}
For an envy-compatible configuration  $(X,R)$, for every $1 \leq \ell \leq |R|$ we define $r_\ell$ as in Equation \eqref{sl} and call $r_\ell$ the representative of group $R_\ell$. Furthermore, we suppose that the groups in $R$ are sorted according to the utility of their representatives, i.e.,
$$
\valu_{r_1}(\alloc_{r_1}) \leq \valu_{r_2}(\alloc_{r_2}) \leq \ldots \leq \valu_{r_{|R|}}(\alloc_{r_{|R|}}).
$$
Note that for an allocation $X$, there might be several different partitions $R$ such that configuration $(X,R)$ is envy-compatible. For example, one trivial such partitions is to put each agent into a separate group. However, 
our interest is in configurations that are more specific. 
In addition to Properties \ref{prop1} and \ref{def-envy-conf} our desired configurations admit another important property. Before we introduce this property, we must define champion and champion-graph. These two concepts were first introduced by Chaudhury \etal~\cite{chaudhury2020efx}. In this paper, we revise their definition to incorporate configurations.

\begin{definition}\label{e-champion}
	Let $(X,R)$ be a configuration  satisfying Property \ref{def-envy-conf} and let $i \in R_\ell$ be an agent. Then, for every subset $S$ such that $v_i(S) > v_{r_\ell}(X_{r_\ell})$, we define $\champset{S}{i}$ to be a  subset of $S$ with minimum number of elements such that $v_i\left(\champset{S}{i}\right)>v_{r_\ell}(X_{r_\ell})$. In case of multiple options for $\champset{S}{i}$, we pick one arbitrarily.
\end{definition}
\begin{observation}\label{S-nonwasteful}
	Let $(X,R)$ be a configuration satisfying Property \ref{def-envy-conf}, $S$ be a subset of goods and $i \in R_\ell$ be an agent such that $v_i(S) > v_{r_\ell}(X_{r_\ell})$. Then, we have $v_i(\champset{S}{i}) = v(\champset{S}{i})$.
\end{observation}
\begin{proof}
	By minimality of $S$, for any good $g \in \champset{S}{i}$, we have $v_i(\champset{S}{i} \setminus\{g\}) \leq v_{r_\ell}(X_{r_\ell})$. Since by definition, we have $v_i(\champset{S}{i})>v_{r_\ell}(X_{r_\ell})$, no good in $\champset{S}{i}$ has value $0$ to agent $i$.
\end{proof}

We now define champion as following.
\begin{definition}[Champion and champion-graph]\label{extended-champ}
	Given a configuration $\sigma = (X, R)$ satisfying Property \ref{def-envy-conf}. We say $i \in R_\ell$ is a champion of $S$ if $v_i(S) > v_{r_\ell}(X_{r_\ell})$ and for every agent $j \neq i$ with $j \in R_k$ and $g \in \champset{S}{i}$ we have 
	$v_{r_k}(X_{r_k}) \geq v_j(\champset{S}{i} \setminus \{g\}).$
	We also define the champion-graph of $\sigma$, denoted by $H_\sigma$ as follows: for every agent $i$ there is a vertex in $H_\sigma$. Edges of $H_\sigma$ are of two types:
	\begin{itemize}
		\item Regular edges: for every pair of agent $i,j $ with $i \in R_\ell$, there is an edge from $i$ to $j$, if $v_i(X_j) > v_{r_\ell}(X_{r_\ell})$. 
		\item Champion edges: for every pair of  agents $i,j$ and every unallocated good $g$, there is a directed edge from $i$ to $j$ with label $g$, if $i$ is a champion of $X_j \cup \{g\}$. 
	\end{itemize}
\end{definition}

The last property we consider for configurations is based on the champion-graph. Note that Property \ref{good-alloc} is only defined on the configurations that 
are envy-compatible.
\begin{property}[Admissibility]\label{good-alloc}
	A configuration $\sigma = (X,R)$ satisfying Property \ref{def-envy-conf} is admissible, if the following properties hold:
	\begin{itemize}
		\item For all $1 \leq \ell \leq |R|$ and every $i \in R_\ell$, there is a path from $r_\ell$ to $i$ in $H_\sigma$ using only regular edges.
		\item For every agents $i \in R_\ell$ and $j \in R_k$ such that $\ell < k$, $v_{r_\ell}(X_{r_\ell}) \geq v_i(X_j)$. I.e., there is no regular edge from $R_\ell$ to $R_k$.
		\item For all $i \notin \{r_1, \ldots, r_{|R|}\}$, we have $v_i(X_i) = v(X_i).$      
	\end{itemize}
\end{property}
We conclude this section by proving Lemma \ref{good-efx}. 

\begin{lemma} \label{good-efx}
	Let $X$ be a non-wasteful allocation, and 
	let $(X,R)$ be a configuration satisfying Properties \ref{def-envy-conf} and \ref{good-alloc}. Then, $X$ is $\EFX$.
\end{lemma}
\begin{proof}
	Consider two agents $i$ and $j$ such that $i \in R_{\ell}$ and $j \in R_k$ for $1 \leq \ell \leq k \leq |R|$. For all $g \in X_i$ we have,
	\begin{align*}
		v_j(X_j) &\geq v_{r_k}(X_{r_k}) &\hbox{$j \in R_k$} \\
		&\geq v_{r_\ell}(X_{r_\ell}) &\hbox{$k \geq \ell$} \\
		&\geq v_i(X_i \setminus \{g\}) &\hbox{$i \in R_\ell$} \\
		&= v(X_i \setminus \{g\}) &\hbox{$X$ is non-wasteful} \\
		&\geq v_j((X_i \setminus \{g\}).
	\end{align*}
	Therefore, $j$ does not strongly envy $i$. 
	If $\ell = k$, with a similar argument, $i$ does not strongly envy $j$ either. 
	
	Now assume $\ell < k$. We have
	\begin{align*}
	v_i(X_i) &\geq v_{r_\ell}(X_{r_\ell}) &\hbox{$i \in R_{\ell}$} \\
	&\geq v_i(X_j). &\hbox{$(X,R)$ is admissible}
	\end{align*}
	Thus, $i$ does not strongly envy $j$. Hence, $X$ is $\EFX$.
\end{proof}

\section{Envy-elimination}
One important part of our algorithm is an auxiliary process which we call \emph{envy-elimination}. This process is designed to bypass possible deadlocks in the updating processes.  This process takes an allocation $X$ as input and returns a configuration $(X', R')$ satisfying Properties \ref{prop1}, \ref{def-envy-conf} and \ref{good-alloc}. To do so, the process first makes the given allocation non-wasteful by simply retaking all the wasted goods. Then, the process chooses an agent with the minimum valuation, i.e.,
$$ \arg \min \valu_i(X_i).$$
 This agent is selected as the representative of set $R_1$ (thus we call her $r_1$). Next, as long as there exists an agent $j \in R_1$ and an agent $i \in N \setminus R_1$ such that $v_{j}(X_i) >  v_{r_1}(X_{r_1})$, we modify the bundle of agent $i$ to $\champset{X_i}{j}$ and add agent $i$ to $R_1$. Note that, the goods in $X_i \setminus \champset{X_i}{j}$ are returned to the pool. When no more agent could be added to $R_1$, we repeat the same process on the agents in $N \setminus R_1$ to construct $R_2$, $R_3$ and so on. Finally, when all the agents are added to some partition $R_i$ we terminate the process.  
 
  Algorithm \ref{envy-elim-alg} shows a pseudocode of the envy-elimination process.  We denote the allocation returned by this process by $X'$ and denote the partition by $R' = \langle R_1,R_2,\ldots, R_{|R'|}\rangle$. Note that, by the way we construct  $X'$, after the process, the valuation of each agent for her bundle  is at most her valuation before the process. However, as we show in Lemma \ref{extended-props}, the output of this process $(X', R')$ satisfies Properties \ref{prop1}, \ref{def-envy-conf} and \ref{good-alloc}.

\begin{algorithm} 
	\caption{Envy-elimination} \label{envy-elim-alg}
	Input :  instance $(N, M, (\util_1, \ldots , \util_n))$, allocation $\allocs$
	
	Output: configuration $(X', R')$ 
	
	\begin{algorithmic}[1]
		\While {$\exists i \in N, g \in X_i$ \suchthat $v_i(g) = 0$}
		\State $X_i \leftarrow X_i \setminus \{g\}$
		\EndWhile
		\State $\ell \leftarrow 1$
		\While {$N \neq \emptyset$}
		\State $r_\ell \leftarrow \argmin_{i \in N}v_i(\alloc_i)$
		\State $X'_{r_\ell} \leftarrow X_{r_\ell}$
		\State $R_\ell \leftarrow \{r_\ell\}$
		\State $N \leftarrow N \setminus \{r_\ell\}$
		\While {$\exists j \in R_\ell, i \in N$ \suchthat $v_j(X_i) > v_{r_\ell}(X_{r_\ell})$}
		\State $X'_i \leftarrow \champset{X_i}{j}$
		\State $R_\ell \leftarrow R_\ell \cup \{i\}$
		\State $N \setminus \{i\}$
		\EndWhile
		\State $\ell \leftarrow \ell+1$
		\EndWhile
		\State Return $(X', \langle R_1, R_2, \ldots, R_{\ell-1} \rangle)$
	\end{algorithmic}	
\end{algorithm}

\begin{lemma}\label{extended-props}
	Let $X$ be an allocation and  let $\sigma = (X', R')$ be the result of running the envy-elimination process on $X$. Then $\sigma$ satisfies Properties \ref{prop1}, \ref{def-envy-conf} and \ref{good-alloc}.
\end{lemma}

\begin{proof} 
	Let $\widetilde{X}$ be the result of removing wasted goods from allocation $X$. Hence, $\widetilde{X}$ is non-wasteful.
	We prove each one of the properties separately.
\paragraph*{Property \ref{def-envy-conf} (Envy-compatibility).} Note that when agent $i$ is added to $R_\ell$, there exists an agent $j$ such that $X'_i = \champset{\widetilde{X}_i}{j}$. 
By Definition \ref{e-champion}, 
\begin{align}
	v_j(X'_i) > v_{r_\ell}(X'_{r_\ell}) \geq v_j(X'_i \setminus \{g\}). \label{eq11}
\end{align}
Moreover,
\begin{align*}
	v_j(X'_i) &= v(X'_i) &\hbox{Observation \ref{S-nonwasteful}}\\
	&= v_i(X'_i) &\hbox{$X'_i \subseteq \widetilde{X}_i$ and $\widetilde{X}$ is non-wasteful}.
\end{align*}
Thus by  Equation \eqref{eq11} we have 
$v_i(X'_i) > v_{r_\ell}(X'_{r_\ell}) \geq v_i(X'_i \setminus \{g\}).$ 
Since we have $v_{r_1}(X'_{r_1}) \leq v_{r_2}(X'_{r_2}) \leq \ldots \leq v_{r_{|R'|}}(X'_{r_{|R'|}})$, $\sigma$ is envy-compatible.

\paragraph{Property \ref{good-alloc} (Admissibility).}
First, we show that for all $1 \leq \ell \leq |R'|$ and all $i \in R_\ell$, there is a path from $r_\ell$ to $i$ in $H_\sigma$ using only regular edges. The proof is by induction on the order that agents are added to $R_\ell$. For $i = r_\ell$, the property holds trivially. Now assume the property holds just before agent $i$ is added to $R_\ell$ meaning that there is a path from $r_\ell$ to $j$. We have
\begin{align*}
	v_j(X'_i) &= v_j(\champset{\widetilde{X}_i}{j}) &\hbox{$X'_i = \champset{\widetilde{X}_i}{j}$}\\
	&> v_{r_\ell}(X'_{r_\ell}) &\hbox{Definition \ref{e-champion}}.
\end{align*}
Therefore, the regular edge $j \rightarrow i$ exists in $H_{\sigma}$ and there is a path from $r_\ell$ to $i$. 

We must also show that for every agents $i \in R_\ell$ and $j \in R_k$ such that $\ell < k$, $v_{r_\ell}(X'_{r_\ell}) \geq v_i(X'_j)$. Assume otherwise. Then there exist $1 \leq \ell < k \leq |R'|$ and agents $i \in R_\ell$ and $j \in R_k$, such that $v_i(X'_j) > v_{r_\ell}(X'_{r_\ell})$.
Thus, before $R_k$ is created, $j$ must be added to $R_\ell$ which is a contradiction. Also, since $\widetilde{X}$ is non-wasteful and for all agents $i$, $X'_i \subseteq \widetilde{X}_i$, $X'$ is also non-wasteful. In particular, for all $i \notin \{r'_1, \ldots r'_{|R'|}\}$, $v_i(X'_i) = v(X'_i)$.
Hence, all properties in Definition \ref{good-alloc} hold and $\sigma$ is admissible.

\paragraph{Property \ref{prop1} ($\EFX$):} Currently, we know that after the envy-elimination process, the resulting configuration $(X',R')$ satisfies Properties \ref{def-envy-conf} and \ref{good-alloc}. Since $X'$ is  non-wasteful, by Lemma \ref{good-efx} we conclude that $X'$ is also $\EFX$.

\end{proof}

\section{An $\EFXX$ Allocation Algorithm}\label{ef2x}
In this section, we present our algorithm for finding a complete $\EFXX$ allocation. Our algorithm consists of four updating rules $\mathsf{U}_0$, $\mathsf{U}_1$, $\mathsf{U}_2$, and $\mathsf{U}_3$.  We start with configuration $\sigma = (X,R)$, where $X$ is an empty allocation and $R = \langle \{1\}, \{2\}, \ldots, \{n\} \rangle$.  At each step, we take the current configuration $\sigma = (X,R)$ as input and update it using one of these rules. These rules are designed in a way that they always (lexicographically) increase the value of $\Phi(\sigma)$, where
$$\Phi(\sigma) = \bigg[ v_{r_1}(X_{r_1}), v_{r_2}(X_{r_2}), \ldots, v_{r_{|R|}}(X_{r_{|R|}}),+\infty, \frac{1}{\sum_{i \notin \{r_1,r_2,\ldots,r_{|R|}\}} |\alloc_{i}|}, \sum_{i \in \{r_1,r_2,\ldots,r_{|R|}\}}  |\alloc_{i}| \bigg].\footnote{ Note that, Since we guarantee that $\sigma$  is envy-compatible, $\Phi(\sigma)$ is well-defined. Recall that $r_1,r_2,\ldots,r_{|R|}$ are the representative agents of partition $R$.}$$

Therefore, after a finite number of updates we obtain a configuration to which none of the updating rules is applicable. Afterwards,
we show that  the remaining goods have a special property that allows us to allocate them to the agents without violating the $\EFXX$ property. In other words, if for an allocation $X$ during the algorithm none of the rules is applicable and goods $ \good_1, \ldots, \good_{|P|}$ are not allocated, then we can find $|P|$ different agents $i_1, i_2, \ldots, i_{|P|}$ such that allocation $\alloc'$ with 
\begin{align*}
\alloc'_{i_\ell} &= \alloc_{i_\ell} \cup \{\good_\ell\} &\mbox{for all $1 \leq \ell \leq |P|$} \\
\alloc'_{j} &= \alloc_{j} &\mbox{for all $j \notin \{i_1, i_2, \ldots, i_{|P|}\}$ }
\end{align*}
is $\EFXX$. 
Algorithm \ref{alg2} shows a pseudocode of our algorithm. In the rest of this section, we first describe the updating rules. For each rule, we prove that the resulting configuration satisfies Properties \ref{prop1}, \ref{def-envy-conf}, and \ref{good-alloc}. Also, we prove that all the rules are $\Phi$-improving. Finally, we describe how we can allocate the remaining goods to maintain the $\EFXX$ property.

Throughout this section, we denote the configuration before and after each update respectively by $\sigma = (X, R)$ and $\sigma' = (X', R')$. In what follows we describe the updating rules and their properties.

\begin{algorithm} 
	\caption{Complete $\EFXX$ allocation} \label{alg2}
	Input : instance $(N, M, (\util_1, \ldots , \util_n))$ 
	
	Output: allocation $X$
	
	\begin{algorithmic}[1]
		\State $\allocs \leftarrow \langle\emptyset, \emptyset, \ldots, \emptyset\rangle$
		\State $R \leftarrow \langle \{1\}, \{2\}, \ldots, \{n\} \rangle$
		\While {$\mathsf{U}_0$ or $\mathsf{U}_1$ or $\mathsf{U}_2$ or $\mathsf{U}_3$ is applicable}
		\State Let $i$ be the minimum index s.t. $\mathsf{U}_i$ is applicable
		\State Update $(X, R)$ via Rule $\mathsf{U}_i$ 
		\EndWhile
		\For {$t \leftarrow |R| \text{ to } 1$}
		\While {$\exists i \in R_t \cap N$ \suchthat $v_i(P)> 0$}
		\State Choose such $i$ with maximum distance from $R_t$
		\State $g \leftarrow \argmax_{h \in P} v_i(h)$
		\State $X_i \leftarrow X_i \cup \{g\}$
		\State $N \leftarrow N \setminus \{i\}$
		\EndWhile
		\EndFor
		\State Let $P_X = \{g_1, g_2, \ldots, g_{|P_X|}\}$
		\State Let $i_1, i_2, \ldots, i_{|P_X|}$ be $|P_X|$ different agents in $N$
		\For {$\ell \leftarrow 1 \text{ to } |P_X|$}
		\State $X_{i_\ell} \leftarrow X_{i_\ell} \cup \{g_\ell\}$
		\EndFor
		\State Return $X$
	\end{algorithmic}	
\end{algorithm}

\subsubsection*{Rule $\mathsf{U}_0$} 
We start with rule $\textsf{U}_0$ which is our most basic rule. This rule allocates an unallocated good $g$ to a representative agent, such that the resulting configuration satisfies Properties  \ref{prop1}, \ref{def-envy-conf} and \ref{good-alloc}.
In the following, you can find a schematic view of Rule $\mathsf{U}_0$. 
\vspace{0.2cm}
\begin{tcolorbox}[colback=blue!5,coltitle=blue!50!black,colframe=blue!25,
title= \textbf{Rule $\mathsf{U}_0$}]
	
	{Preconditions:}\label{pre1}
	\begin{itemize}
		\item $\sigma = (X, R)$ satisfies Properties \ref{prop1}, \ref{def-envy-conf}  and \ref{good-alloc}.
		\item There exists an unallocated good $ \good \in P$ and an index $\ell$ such that $v_{r_\ell}(g) = 0$ and for every $1 \leq k \leq |R|$, 	$$\forall_{i \in R_k} \quad v_i(X_{r_\ell} \cup \{g\}) \leq v_{r_{k}}(X_{r_{k}}).$$ 
	\end{itemize}
	\hrulefill
	
	{Process:}
	\begin{itemize}
		\item Allocate $\good$ to $r_\ell$. 
		\item Set $R' = R$.
	\end{itemize}
	\hrulefill
	
	{Guarantees:}
	\begin{itemize}
		\item $\sigma' = (X', R')$ satisfies Properties \ref{prop1}, \ref{def-envy-conf}, and \ref{good-alloc}. 
		\item $\Phi(\sigma') {\lexlarger} \Phi(\sigma)$.	
	\end{itemize}
\end{tcolorbox}

\begin{observation}\label{same-util}
	Let $\sigma' = (X', R')$ be the result of updating $\sigma = (X, R)$ via Rule $\textsf{U}_0$. Then for every agent $i$ we have $v_i(X'_i) = v_i(X_i)$.
\end{observation}
\begin{lemma}\label{good-alloc-0}
	Let $\sigma' = (X', R')$ be the result of updating $\sigma = (X, R)$ via Rule $\textsf{U}_0$. Then, $\sigma'$ satisfies Properties \ref{prop1}, \ref{def-envy-conf} and \ref{good-alloc}. 
\end{lemma}
\begin{proof} 
	We prove each one of the properties separately. 
	\paragraph{Property \ref{prop1} ($\EFX$):}
	In order to prove that $X'$ is $\EFX$, it suffices to prove that no agent strongly envies $X'_{r_\ell}$. For all $1 \leq k \leq |R|$ and $i \in R_k$, 
	\begin{align*}
		v_i(X'_i) &= v_i(X_i) &\hbox{Observation \ref{same-util}}\\
		&\geq v_{r_k}(X_{r_k}) &\hbox{$i \in R_k$}\\
		&\geq v_i(X_{r_\ell} \cup \{g\}). &\hbox{Precondition of Rule $\mathsf{U}_0$}
	\end{align*}

	Thus, no agent envies $r_\ell$. 
	
	\paragraph{Property \ref{def-envy-conf} (Envy-compatibility):}
	By Observation \ref{same-util}, for all $1 \leq k \leq |R|$, 
	$$r_k = \arg \min_{j \in R_k} \valu_j(X'_j).$$
	Also, for all $1 \leq k \leq |R|, i \in R_k, i \neq r_\ell$ and $g \in X_i$ we have
	\begin{align*}
		v_i(X'_i \setminus \{\good\}) &= v_i(X_i \setminus \{\good\}) &\hbox{$X'_i = X_i$}\\
		&\leq v_{r_k}(X_{r_k}) &\hbox{$\sigma$ is envy-compatible}\\
		&= v_{r_k}(X'_{r_k}) &\hbox{Observation \ref{same-util}}.
	\end{align*}

	\paragraph*{Property \ref{good-alloc} (Admissibility):} 
	Note that after the update, agent $r_\ell$ has received one additional good and all other bundles are untouched. Since $v_{r_\ell}(X'_{r_\ell}) = v_{r_\ell}(X_{r_\ell})$, the only difference between $H_{\sigma'}$ and $H_\sigma$ is that $r_\ell$ might have additional incoming edges. Also, since $\sigma$ is admissible, for all $1 \leq \ell \leq |R|$ and all $i \in R_\ell$, there is a path from $r_\ell$ to $i$ in $H_{\sigma}$ using only regular edges. Hence, for all $1 \leq \ell \leq |R|$ and all $i \in R_\ell$, there is a path from $r_\ell$ to $i$ in $H_{\sigma'}$ which uses only regular edges.
	
	Now we prove that for every agents $i \in R_k$ and $j \in R_{k'}$ such that $k < k'$, $v_{r_k}(X'_{r_k}) \geq v_i(X'_j)$. If $j = r_\ell$, we have
	\begin{align*}
		v_{r_k}(X'_{r_k}) &= v_{r_k}(X_{r_k}) &\hbox{Observation \ref{same-util}}\\
		&\geq v_i(X_{r_\ell} \cup \{g\}) &\hbox{Precondition of Rule $\mathsf{U}_0$}\\
		&= v_i(X'_j). &\hbox{$j = r_\ell$}
	\end{align*} 
	Also, if $j \neq r_\ell$, we have
	\begin{align*}
		v_{r_k}(X'_{r_k}) &= v_{r_k}(X_{r_k}) &\hbox{Observation \ref{same-util}}\\
		&\geq v_i(X_j) &\hbox{$\sigma$ is admissible}\\
		&= v_i(X'_j). &\hbox{$X'_j = X_j$}
	\end{align*} 
	Therefore, for every agents $i \in R_k$ and $j \in R_{k'}$ such that $k < k'$, $v_{r_k}(X'_{r_k}) \geq v_i(X'_j)$. Also, for all $i \notin \{r_1, \ldots, r_{|R|}\}$, $X'_i = X_i$. By admissibility of $X$, $v_i(X_i) = v(X_i)$ and hence, $v_i(X'_i) = v(X'_i)$. Thus, Property \ref{good-alloc} holds.
\end{proof}
\begin{lemma}\label{phi-imp-0}
	Rule $\textsf{U}_0$ is $\Phi$-improving.
\end{lemma}
\begin{proof}
Let $\{r_1, r_2, \ldots, r_{|R|}\}$ be the set of representatives before the update. 
	We have
	\begin{align*}
		\Phi(\sigma') &= \bigg[ v_{r_1}(X'_{r_1}), \ldots, v_{r_{|R|}}(X'_{r_{|R|}}),+\infty, \frac{1}{\sum_{i \notin A} |X'_{i}|}, \sum_{i \in A}  |X'_{i}| \bigg] \\
		&= \bigg[ v_{r_1}(X_{r_1}), \ldots, v_{r_{|R|}}(X_{r_{|R|}}),+\infty, \frac{1}{\sum_{i \notin A} |X'_{i}|}, \sum_{i \in A}  |X'_{i}| \bigg] &\hbox{Observation \ref{same-util}} \\
	\intertext{Since for all $i \neq r_\ell$ we have $X'_i = X_i$  and $X'_{r_\ell} = X_{r_\ell} \cup \{g\}$, we have}
		\Phi(\sigma') &= \bigg[ v_{r_1}(X_{r_1}), \ldots, v_{r_{|R|}}(X_{r_{|R|}}),+\infty, \frac{1}{\sum_{i \notin A} |X_{i}|}, \sum_{i \in A}  |X_{i}|+1 \bigg]\\
		&\lexlarger \bigg[ v_{r_1}(X_{r_1}), \ldots, v_{r_{|R|}}(X_{r_{|R|}}),+\infty, \frac{1}{\sum_{i \notin A} |\alloc_{i}|}, \sum_{i \in A}  |\alloc_{i}| \bigg] \\
		&= \Phi(\sigma).
	\end{align*}
\end{proof}
\subsubsection*{Rule $\mathsf{U}_1$}

Our second updating rule updates the  configuration using a special type of cycle in the champion-graph called champion cycle. 

\begin{definition}\label{extended-champion cycle}
	For a configuration $\sigma = (X, R)$, we call a cycle $C$ of $H_\sigma$ a \textit{champion cycle}, if for every champion edge $i \xrightarrow[]{g} j$ in $C$, $j \in \{r_1, r_2, \ldots, r_{|R|}\}$.
	Additionally, for every two champion edges of $C$ with labels $g$ and $g'$, we have that $g\neq g'$.
\end{definition}
In this rule, we first find a \textit{champion cycle} in $H_\sigma$ and update the allocation through this cycle. Next, we apply envy-elimination.
In the following, you can find a schematic view of Rule $\mathsf{U}_1$.

\vspace{0.5cm}
\begin{tcolorbox}[colback=blue!5,coltitle=blue!50!black,colframe=blue!25,
	title= \textbf{Rule $\mathsf{U}_1$}]
	
	{Preconditions:}
	\begin{itemize}
		\item $\sigma = (X, R)$ satisfies Properties \ref{prop1}, \ref{def-envy-conf}  and \ref{good-alloc}.
		\item There exists a \textit{champion cycle} $C$ in $H_\sigma$.
	\end{itemize}
	\hrulefill
	
	{Process:}
	\begin{itemize}
		\item For every edge $i \rightarrow j\in C$, set $\widetilde{X}_i = X_j$.
		\item For every champion edge $i \xrightarrow[]{g} j \in C$, set $\widetilde{X}_i = \champset{X_j \cup \{g\}}{i}$.
		\item For every agent $j \notin C$, set $\widetilde{X}_j = X_j$.
		\item Apply envy-elimination on $\widetilde{X}$.
	\end{itemize}
	\hrulefill

	{Guarantees:}
	\begin{itemize}
		\item $\sigma' = (X', R')$ satisfies Properties \ref{prop1}, \ref{def-envy-conf}, and \ref{good-alloc}.
		\item $\Phi(\sigma') {\lexlarger} \Phi(\sigma)$.
	\end{itemize}
	
\end{tcolorbox}
\vspace{0.5cm}

Since we apply envy-elimination in the last step of Rule $\mathsf{U}_1$, Lemma \ref{extended-props} implies that the resulting configuration satisfies Properties \ref{prop1}, \ref{def-envy-conf}, and \ref{good-alloc}.

\begin{corollary}[of Lemma \ref{extended-props}]\label{good-alloc-2}
	Let $\sigma' = (X', R')$ be the result of updating $\sigma = (X, R)$ via Rule $\textsf{U}_1$. Then, $\sigma' = (X', R')$ satisfies Properties \ref{prop1}, \ref{def-envy-conf}, and \ref{good-alloc}.
\end{corollary}

\begin{lemma}\label{phi-imp-2}
	Rule $\textsf{U}_1$ is $\Phi$-improving.
\end{lemma}
\begin{proof}
	We prove that the conditions of Lemma \ref{main-Phi-improving} hold.
	Let $\ell$ be the smallest index such that $r_\ell \in C$ and $\widetilde{X}$ be the allocation just before applying envy-elimination. 
	Then we have, 
	\begin{align*}
		\hbox{$\forall i \in R_1 \cup \ldots \cup R_{\ell-1}, \quad \widetilde{X}_i = X_i$.}
	\end{align*}
	Therefore, Condition \ref{cond1} of Lemma \ref{main-Phi-improving} holds.
	
	Without loss of generality, we can assume for all $k > \ell$ such that $r_k \notin C$, $v_{r_k}(X_{r_k}) > v_{r_\ell}(X_{r_\ell})$.\footnote{Otherwise, we can simply reorder the groups in $R$ which their representative agents have value $v_{r_\ell}(X_{r_\ell})$.} Assume $i \in R_k$ and $k \geq \ell$.
	If $i \notin C$, we have
	\begin{align*}
		v_i(\widetilde{X}_i) &= v_i(X_i) &\hbox{$\widetilde{X}_i = X_i$}\\ 
		&\geq v_{r_k}(X_{r_k}) &\hbox{$i \in R_k$}\\
		&> v_{r_\ell}(X_{r_\ell}).
	\end{align*}  
	If $i \rightarrow j \in C$, we have
	\begin{align*}
		v_i(\widetilde{X}_i) &= v_i(X_j) &\hbox{$\widetilde{X}_i = 	X_j$}\\
		&> v_{r_k}(X_{r_k}) &\hbox{Regular edge from $i$ to $j$}\\
		&> v_{r_\ell}(X_{r_\ell}).
	\end{align*}
	If $i \xrightarrow[]{g} j \in C$, we have
	\begin{align*}
		v_i(\widetilde{X}_i) &= v_i(\champset{X_j \cup \{g\}}{i}) &\hbox{$\widetilde{X}_i = 	\champset{X_j \cup \{g\}}{i}$}\\
		&> v_{r_k}(X_{r_k}) &\hbox{Definition \ref{e-champion}}\\
		&> v_{r_\ell}(X_{r_\ell}).
	\end{align*}
	 Therefore, we have 
	\begin{align*}
		\forall i \in R_\ell \cup \ldots \cup R_{|R|}, \quad v_i(\widetilde{X}_i) > v_{r_\ell}(X_{r_\ell}).
	\end{align*}
	Thus, Condition \ref{cond2} of Lemma \ref{main-Phi-improving} holds too and therefore, $\Phi(\sigma') \lexlarger \Phi(\sigma)$.
\end{proof}

In Lemmas \ref{champ-edge-exists} and \ref{small-pool} we prove two properties of every admissible configuration $\sigma$ on which neither of the Rules $\mathsf{U}_0$ or $\mathsf{U}_1$ is applicable. We use these properties in Section \ref{correctness-2} to obtain a complete $\EFXX$ allocation.
\begin{lemma}\label{champ-edge-exists}
	For an admissible configuration $\sigma = (X, R)$, if neither of the Rules $\textsf{U}_0$ and $\textsf{U}_1$ is applicable, then for every $r_\ell \in \{r_1, r_2, \ldots, r_{|R|}\}$ and every unallocated good $g$ there exists an agent $j \notin R_\ell$ such that $j$ is a champion of $X_{r_\ell} \cup \{g\}$.
\end{lemma}
\begin{proof}
	Assume otherwise.
	If $v_{r_\ell}(g) = 0$, then Rule $\mathsf{U}_0$ is applicable which is a contradiction. Otherwise, $v_{r_\ell}(g) = v(g)$ and $v_{r_\ell}(X_{r_\ell} \cup \{g\}) > v_{r_\ell}(X_{r_\ell})$. If for an agent $j \in R_\ell$, $j$ is a champion of $X_{r_\ell} \cup \{g\}$, then $H_\sigma$ has a champion cycle and Rule $\mathsf{U}_1$ is applicable. If for no other agent $j \neq r_\ell$, $j$ is a champion of $X_{r_\ell} \cup \{g\}$, then $r_\ell$ is a champion of $X_{r_\ell} \cup \{g\}$ and again Rule $\mathsf{U}_1$ is applicable. Therefore, there always exists an agent $j \notin R_\ell$ such that the $j$ is a champion of $X_{r_\ell} \cup \{g\}$. 
\end{proof}

\begin{lemma}\label{small-pool}
	For an admissible configuration $\sigma = (X, R)$, if neither of the Rules $\textsf{U}_0$ and $\textsf{U}_1$ is applicable, then $|P_X| < |R|$.
\end{lemma}
\begin{proof}
	Let $g_1$ be an arbitrary good in the pool and let $j_1 \in R_{t_1}$ be an arbitrary agent in $\{r_1, r_2, \ldots, r_{|R|}\}$. By Lemma \ref{champ-edge-exists}, there exists an agent $j'_1 \in R_{t_2} \neq R_{t_1}$ such that the edge $j'_1 \xrightarrow[]{g} j_1$ exists in $H_\sigma$. Let $j_2$ be the agent in $R_{t_2} \cap \{r_1, r_2, \ldots, r_{|R|}\}$. Let $g_2 \neq g_1$ be another good in the pool and let $j'_2 \in R_{t_3}$ for $t_3 \notin \{t_1, t_2\}$ be the agent such that the edge $j'_2 \xrightarrow[]{g_2} j_2$ exists in $H_\sigma$. 
	We continue repeating this process to obtain the sequence $j_1, j_2, \ldots, j_{|P|}$ such that for all $1 \leq \ell \leq |P|$, $j'_\ell \in R_{t_{\ell+1}}$. Note that since Rule $\mathsf{U}_1$ is not applicable, there is no champion cycle in $H_\sigma$ and hence for $\ell \neq \ell'$, we have $t_\ell \neq t_{\ell'}$. Therefore, $R$ should partition agents into at least $|P|+1$ different non-empty parts. Hence, $|P|<|R|$.
\end{proof}
\subsubsection*{Rule $\mathsf{U}_2$} 
In this updating rule, we check if there exists a representative agent such that her value for at least one of the remaining goods is non-zero. If so, we simply allocate it to that representative agent and apply envy-elimination. 
In the following, you can find a schematic view of Rule $\mathsf{U}_2$.

\vspace{0.5cm}
\begin{tcolorbox}[colback=blue!5,coltitle=blue!50!black,colframe=blue!25,
	title= \textbf{Rule $\mathsf{U}_2$}]
	
		Preconditions:\label{pre4}
	\begin{itemize}
		\item $\sigma = (X, R)$ satisfies Properties \ref{prop1}, \ref{def-envy-conf}  and \ref{good-alloc}.
		\item There exists a representative agent $r_\ell$ and an unallocated good $g$ such that $v_{r_\ell}(g) = v(g)$.
		\item Rules $\mathsf{U}_0$ and $\mathsf{U}_1,$ are not applicable.
	\end{itemize}
	\hrulefill
	
	{Process:}
	\begin{itemize}
		\item Allocate $g$ to $r_\ell$.
		\item Apply envy-elimination.
	\end{itemize}
	\hrulefill
	
	{Guarantees:}
	\begin{itemize}
		\item $\sigma' = (X', R')$ satisfies Properties \ref{prop1}, \ref{def-envy-conf}, and \ref{good-alloc}.
		\item $\Phi(\sigma') {\lexlarger} \Phi(\sigma)$.	
	\end{itemize}
\end{tcolorbox}
\vspace{0.5cm}

Since we apply envy-elimination in the last step of Rule $\mathsf{U}_2$, Lemma \ref{extended-props} implies that the resulting configuration satisfies Properties \ref{prop1}, \ref{def-envy-conf}, and \ref{good-alloc}.

\begin{corollary}[of Lemma \ref{extended-props}]\label{good-alloc-3}
	Let $\sigma' = (X', R')$ be the result of updating $\sigma = (X, R)$ via Rule $\textsf{U}_2$. Then, $\sigma' = (X', R')$ satisfies Properties \ref{prop1}, \ref{def-envy-conf}, and \ref{good-alloc}.
\end{corollary}

\begin{lemma}\label{phi-imp-3}
	Rule $\textsf{U}_2$ is $\Phi$-improving.
\end{lemma}
\begin{proof}
	We prove that the conditions of Lemma \ref{main-Phi-improving} hold.
	We have, 
	\begin{align*}
	\hbox{$\forall i \in R_1 \cup \ldots \cup R_{\ell-1}, \quad \widetilde{X}_i = X_i$.}
	\end{align*}
	Therefore, Condition \ref{cond1} of Lemma \ref{main-Phi-improving} holds.
	
	Without loss of generality, we can assume for all $k > \ell$, $v_{r_k}(X_{r_k}) > v_{r_\ell}(X_{r_\ell})$.\footnote{Otherwise, we can simply reorder the groups in $R$ which their representative agents have value $v_{r_\ell}(X_{r_\ell})$.} 
	For all $k > \ell$ and $i \in R_k$, we have 
	\begin{align*}
		v_i(\widetilde{X}_i)& = v_i(X_i) &\hbox{$\widetilde{X}_i = X_i$}\\
		&\geq v_{r_k}(X_{r_k}) &\hbox{$i \in R_k$} \\
		&> v_{r_\ell}(X_{r_\ell}).
	\end{align*}
	For all $i \in R_\ell \setminus r_\ell$, we have 
	\begin{align*}
		v_i(\widetilde{X}_i) &= v_i(X_i) &\hbox{$\widetilde{X}_i = X_i$}\\
		&> v_{r_\ell}(X_{r_\ell}). &\hbox{$i \in R_\ell \setminus \{r_\ell\}$}
	\end{align*}
	Finally we have 
	\begin{align*}
		v_{r_\ell}(X'_{r_\ell}) &= v_{r_\ell}(X_{r_\ell}) + v_{r_\ell}(g) &\hbox{$X'_{r_\ell} = X_{r_\ell} \cup \{g\}$} \\
		&> v_{r_\ell}(X_{r_\ell}). &\hbox{$v_{r_\ell}(g) = v(g) > 0$}
	\end{align*}
	Therefore,
	\begin{align*}
	\hbox{$\forall i \in R_\ell \cup \ldots \cup R_{|R|}, \quad v_i(\widetilde{X}_i) > v_{r_\ell}(X_{r_\ell}).$}
	\end{align*}
	Thus, Condition \ref{cond2} of Lemma \ref{main-Phi-improving} holds too and hence, $\Phi(\sigma') \lexlarger \Phi(\sigma)$.
\end{proof}

\subsubsection*{Rule $\mathsf{U}_3$} 
This updating rule is structurally different from Rules $\mathsf{U}_0$, $\mathsf{U}_1$, and $\mathsf{U}_2$. Unlike the previous rules, in Rule $\mathsf{U}_3$ the agent to whom we allocate a good is not a representative in $\sigma$. Here we find an unallocated good $g$ and a non-representative agent $i$ such that $v_i(g) = v(g)$ and allocating $g$ to agent $i$ violates the $\EFXX$ property. Moreover, we require that if $i \in R_\ell$, some agent $i' \in R_{\ell'}$ for $\ell' \leq \ell$ envies $X_i \cup \{g\}$ even after removal of two goods. 
Basically, $i' \in R_{\ell'}$ for $\ell' \leq \ell$ is a reason for the new allocation to not be $\EFXX$. Then, we allocate $g$ to $i$ and perform envy-elimination on $X$. 
In the following, you can find a schematic view of Rule $\mathsf{U}_3$.

\vspace{0.5cm}
\begin{tcolorbox}[colback=blue!5,coltitle=blue!50!black,colframe=blue!25,
	title= \textbf{Rule $\mathsf{U}_3$}]
	
	{Preconditions:}\label{pre1}
	\begin{itemize}
		\item $\sigma = (X, R)$ satisfies Properties \ref{prop1}, \ref{def-envy-conf}  and \ref{good-alloc}.
		\item Rules $\mathsf{U}_0$, $\mathsf{U}_1$, and $\mathsf{U}_2$ are not  applicable.
		\item There exists $i \in R_\ell \setminus \{r_\ell\}$ with distance $d$ from $r_\ell$ in $H_\sigma$ and another $i' \in R_{\ell'}$ with distance $d'$ from $r_{\ell'}$ in $H_\sigma$ such that $(\ell, d) \lexlargereq (\ell', d')$ and there exists $g \in P$ such that $v_i(g) = v(g)$ and $i'$ envies $X_i \cup \{g\}$ even after removal of $2$ goods.
	\end{itemize}
	\hrulefill
	
	{Process:}
	\begin{itemize}
		\item Allocate $g$ to $i$.
		\item Apply envy-elimination on $X$.
	\end{itemize}
	\hrulefill
	
	{Guarantees:}
	\begin{itemize}
		\item $\sigma' = (X', R')$ satisfies Properties \ref{prop1}, \ref{def-envy-conf}, and \ref{good-alloc}.
		\item $\Phi(\sigma') {\lexlarger} \Phi(\sigma)$.	
	\end{itemize}
	
\end{tcolorbox}
\vspace{0.5cm}

For the rest of this section, assume $(\ell', d')$ is lexicographically smallest such that for some $i' \in R_{\ell'}$ with distance $d'$ from $r_{\ell'}$, $i'$ envies $X_i \cup \{g\}$ even after removal of $2$ goods.

Since we apply envy-elimination in the last step of Rule $\mathsf{U}_3$, Lemma \ref{extended-props} implies that the resulting configuration satisfies Properties \ref{prop1}, \ref{def-envy-conf}, and \ref{good-alloc}.

\begin{corollary}[of Lemma \ref{extended-props}]\label{good-alloc-4}
	Let $\sigma' = (X', R')$ be the result of updating $\sigma = (X, R)$ via Rule $\textsf{U}_3$. Then, $\sigma' = (X', R')$ satisfies Properties \ref{prop1}, \ref{def-envy-conf}, and \ref{good-alloc}.
\end{corollary}

In order to prove that Rule $\mathsf{U}_3$ is $\Phi$-improving, first we prove Lemma \ref{subset-sources-4}.

\begin{lemma}\label{subset-sources-4}
	Let $\sigma' = (X', R')$ be the result of updating $\sigma = (X, R)$ via Rule $\textsf{U}_3$. Then, for every $1 \leq k \leq |R|$ and every agent $j \in R_k$, 
	\begin{itemize}
		\item if $j \in R'_t$, then $v_{r'_t}(X'_{r'_t}) \leq v_{r_k}(X_{r_k})$, and
		\item if $j \notin \{r_1, r_2, \ldots, r_{|R|}\}$, then $j \notin \{r'_1, r'_2, \ldots, r'_{|R'|}\}$. In other words, $\{r'_1, r'_2, \ldots, r'_{|R'|}\} \subseteq \{r_1, r_2, \ldots, r_{|R|}\}$.
	\end{itemize}
\end{lemma} 
\begin{proof} 
	The proof is by induction on $(k, d)$ where $d$ is the distance of $j$ from $r_k$ for $j \in R_k$. We have $(k_1, d_1) < (k_2, d_2)$ if either $k_1 < k_2$ or $k_1 = k_2$ and $d_1 < d_2$. 
	
	\begin{description}\setlength{\itemsep}{0pt}
		\item [Case 1 $j \neq i$:] If $j = r_k$, then either $j \in \{r'_1, \ldots, r'_{|R'|}\}$ and $v_{r'_t}(X'_{r'_t}) = v_{r_k}(X_{r_k})$ or, $v_{r'_t}(X'_{r'_t}) < v_{r_k}(X'_{r_k}) \leq v_{r_k}(X_{r_k})$.
		
		Now assume $j \neq r_k$, and $j'$ is the preceding node of $j$ in a shortest path from $r_k$ to $j$ in $H_\sigma$. By induction assumption, $j' \in R'_t$ for $v_{r'_t}(X'_{r'_t}) \leq v_{r_k}(X_{r_k})$. 
		The bundle of $j$ is not changed before the envy-elimination. Also since $X$ is admissible and $j$ is not representative, $X_j$ has no wasted good. Therefore, the bundle of $j$ is not changed even after the removal of wasted goods in the envy-elimination process. We have $v_j(X_j) > v_{r_k}(X_{r_k}) \geq v_{r'_t}(X'_{r'_t})$ and $v_{j'}(X_j) = v(X_j)$. Therefore, either $j$ is already added to some $R'_{t'}$ for $t' \leq t$ or $j$ gets added to $R'_t$.
		\item [Case 2 $j = i$:] There is an agent in $\{R_1 \cup \ldots \cup R_\ell\}$ namely $i' \in R'_{t'}$ such that 
		\begin{align*}
			v_{i'}(X'_i) &> v_{i'}(X_{i'}) &\hbox{$i'$ strongly envies $X'_i$}\\
			&\geq v_{i'}(X'_{i'}) &\hbox{$X'_{i'} \subseteq X_{i'}$} \\
			&\geq v_{r'_{t'}}(X_{r'_{t'}}). &\hbox{$i' \in R'_{t'}$}
		\end{align*}
		Now consider the moment $i'$ is added to $R'_{t'}$. We know, 
		\begin{align*}
			v_i(X_i \cup \{g\}) &= v(X_i \cup \{g\}) &\hbox{$X$ is admissible and $v_i(g) = v(g)$}\\
			&\geq v_{i'}(X_i \cup \{g\}) &\hbox{$v_{i'}(g) > 0$ by Lemma \ref{non-waste-good}}\\
			&> v_{i'}(X_{i'}). 
		\end{align*}
		Therefore, $i$ could not be considered for being a representative before $i'$ is added to $R'_{t'}$. Thus, either $i$ is already added to some $R'_t$ for $t \leq t'$ or $i$ gets added to $R'_{t'}$ and $t = t'$. Hence, we have
		\begin{align*}
			v_{r'_t}(X'_{r'_t}) &\leq v_{r'_{t'}}(X'_{r'_{t'}}) &\hbox{$t \leq t'$}\\
			&\leq v_{r_{\ell'}}(X_{r_{\ell'}}) &\hbox{Induction hypothesis}\\
			&\leq v_{r_\ell}(X_{r_\ell}). &\hbox{$\ell' \leq \ell$}
		\end{align*} 
	\end{description}
\end{proof}
\begin{lemma}\label{phi-imp-4}
	Rule $\textsf{U}_3$ is $\Phi$-improving.
\end{lemma}
\begin{proof}
	Let $A = \{r_1, r_2, \ldots, r_{|R|}\}$ and $A' = \{r'_1, r'_2, \ldots, r'_{|R'|}\}$.
	By Lemma \ref{subset-sources-4}, $A' \subseteq A$. Also note that for all the representative agents $r'_k$ in $\sigma'$, $X'_{r'_k} \geq X_{r'_k}$. If $A' \subsetneq A$, by Lemma \ref{lex1}, $\Phi(\sigma') \lexlarger \Phi(\sigma)$. Otherwise, $A' = A$. Note that 
	\begin{align*}
		|X_i| > |X_i \cup \{g\}|-2 \geq |X'_i|  ,
	\end{align*}
	and for all agents $j \neq i$, 
	\begin{align*}
	|X_j| \geq |X'_j|.
	\end{align*}
	Therefore, 
	\begin{align}
		\sum_{j \notin A} |X_j| > \sum_{j \notin A'} |X'_j|. \label{ineq2}
	\end{align}
	Hence,
	\begin{align*}
		\Phi(\sigma') &= \bigg[ v_{r'_1}(X'_{r'_1}), v_{r'_2}(X'_{r'_2}), \ldots, v_{r'_{|R'|}}(X'_{r'_{|R'|}}),+\infty, \frac{1}{\sum_{j \notin A'} |X'_j|}, \sum_{j \in A'} |X'_j| \bigg] \\
		&= \bigg[ v_{r_1}(X_{r_1}), v_{r_2}(X_{r_2}), \ldots, v_{r_{|R|}}(X_{r_{|R|}}),+\infty, \frac{1}{\sum_{j \notin A'} |X'_j|}, \sum_{j \in A'} |X'_j| \bigg] \\
		&\lexlargereq \bigg[ v_{r_1}(X_{r_1}), v_{r_2}(X_{r_2}), \ldots, v_{r_{|R|}}(X_{r_{|R|}}), +\infty, \frac{1}{\sum_{j \notin A} |X_j|}, \sum_{j \in A} |X_j| \bigg] &\hbox{Inequality \eqref{ineq2}}\\
		&= \Phi(\sigma).
	\end{align*}
\end{proof}

\subsection{Allocating the Remaining Goods}\label{correctness-2}

As we mentioned, our algorithm continues to update the allocation as long as one of the rules  is applicable. Note that, since there are finitely many possible allocations and the potential function $\Phi$ increases after each update, we eventually end up with a configuration $\sigma = (X, R)$ such that none of the rules can be applied on $\sigma$. 
At that moment, if allocation $X$ allocates all the goods, we are done. Otherwise, there are some goods remaining in the pool. 
In this section, we consider the latter case; we show how to allocate these goods so that the final allocation preserves the $\EFXX$ property. We denote the configuration at this step by $\sigma = (X, R)$ and the final complete allocation by $X'$.

Generally, our strategy is to find  $|P|$ different agents and allocate one good to each one of them. Recall that by Lemma \ref{small-pool}, we have $|N| > |P|$ and therefore selecting $|P|$ different agents from $N$ is possible. The process of allocating the remaining goods is illustrated in Algorithm \ref{alg2}. In this algorithm, 
we start by $\ell=|R|$ and as long as there is an agent $i \in R_\ell \cap N$ such that $v_i(P) > 0$, we give $i$ with maximum distance from $r_\ell$ her most desirable good from the pool\footnote{If there were multiple such goods, we select an arbitrary one} and remove $i$ from $N$. Then, we decrease $\ell$ by one and repeat the same process. The process goes on until either no good remains in the pool, or $\ell=0$. Let us call the set of agents that  receive one good in this process by $N_1$. 
Let us denote the set of unallocated goods at this stage by $M_2$.
In case that $\ell=0$ and the pool is still not empty (i.e., $M_2 \neq \emptyset$), we allocate each remaining good to an arbitrary agent in $N \setminus N_1$ (one good to each agent). We denote the agents that receive some good in $M_2$ by $N_2$. 

Let $X'$ be the allocation after this step. By definition, we know that $X'$ is complete. In the rest of this section,  we prove Theorem \ref{main-theorem-2} which shows that $X'$ is also $\EFXX$.

\begin{theorem}\label{main-theorem-2}
	Assuming that the valuation functions are restricted additive, Algorithm \ref{alg2} returns a complete $\EFXX$ allocation.
\end{theorem}

Note that
by Observation \ref{good-alloc-0} and Corollaries \ref{good-alloc-2}, \ref{good-alloc-3}, and \ref{good-alloc-4}, at the time when none of the updating rules is applicable, $X$ is $\EFX$. Now we prove that after allocating the remaining goods by Algorithm \ref{alg2}, the resulting allocation $X'$ is $\EFXX$. 

We start by a simple observation.
\begin{observation}\label{wasted-goods}
	For all $i \in N \setminus N_1$ and $g \in M_2$, we have $v_i(g) = 0$.
\end{observation}

Since the partial allocation before the last step was $\EFX$, and no good is allocated to the agents in $N \setminus (N_1 \cup N_2)$ in the last step, to show that the final allocation is $\EFXX$ it is enough to prove that no agent envies any agent $j \in N_1 \cup N_2$ after removal of any two goods from the bundle of agent $j$. 

As a contradiction suppose that  for an agent $i$ and some agent $j \in N_1 \cup N_2$, there exists two goods $g_1, g_2 \in X'_j$ such that agent $i$ envies the bundle of agent $j$ after eliminating goods $g_1$ and $g_2$.
Let $g$ be the last good which is allocated to agent $j$ (via Algorithm \ref{alg2}).
By Observation \ref{non-waste-good}, we have $v_i(g) = v(g)$. 
Let $\ell$ and $\ell'$ be such that $i \in R_\ell$ and $j \in R_{\ell'}$. Let $d$ be the distance of $i$ from $r_\ell$ and $d'$ be the distance of $j$ from $r_{\ell'}$ in $H_\sigma$. If $j \in N_1$ and $j$ is a representative, then the preconditions of Rule $\mathsf{U}_2$ hold for $\sigma$. If $j \in N_1$ and $(\ell', d') \lexlargereq (\ell,d)$, then the preconditions of Rule $\mathsf{U}_3$ hold for $\sigma$. Therefore in either of these cases, Rule $\mathsf{U}_2$ or Rule $\mathsf{U}_3$ is applicable on $\sigma$ which is a contradiction. Thus, we have $(\ell, d) \lexlarger (\ell', d')$ or $j \in N_2$. This means that in the last stage of the algorithm, some good $g'$ should be allocated to $i$ before $g$ is allocated to $j$ with the property that
	\begin{align}
	v_i(g') &\geq v_i(g). \label{ineq1}
	\end{align}
	Without loss of generality, we can assume $g \neq g_1$.
	Thus, we have,
	\begin{align*}
	v_i(X'_j \setminus \{g_1, g_2\}) &= v_i((X_j \cup \{g\}) \setminus \{g_1, g_2\})  \\
	&= v_i(X_j \setminus \{g_1\}) + v_i(g) - v_i(g_2) \\
	&\leq v_i(X_j \setminus \{g_1\}) + v_i(g') - v_i(g_2) &\hbox{Inequality \eqref{ineq1}}\\
	&\leq v_i(X_i) + v_i(g') &\hbox{$X$ is $\EFX$}\\
	&=v_i(X'_i),
	\end{align*}
	which contradicts our assumption that agent $i$ envies $X'_j \setminus \{g_1, g_2\}$. This completes the proof.

\section{An $\EFX$ Allocation with at most $\floor{n/2}-1$ Discarded Goods}\label{efxn2}
In this section we present an algorithm that leaves aside at most $\floor{n/2}-1$ many goods and finds an $\EFX$ allocation for the rest of the goods. Algorithm \ref{alg1} shows a pseudocode of our allocation method. 
Similar to Algorithm \ref{alg2}, Algorithm \ref{alg1} starts with the empty allocation $X$ and partition $R = \langle \{1\}, \{2\}, \ldots, \{n\} \rangle$. Our algorithm consists of four updating rules $\mathsf{U}_0$, $\mathsf{U}_1$, $\mathsf{U}_2$ and $\mathsf{U}_4$. At each step, we take the current configuration as input and update it using one of these updating rules. 
Rule $\mathsf{U}_0$, $\mathsf{U}_1$, and $\mathsf{U}_2$ are already defined in Section \ref{ef2x}.
Similar to these rules, Rule $\mathsf{U}_4$ is designed in a way that the value of $\Phi(\sigma)$ increases after each update. When none of the rules $\mathsf{U}_0$, $\mathsf{U}_1$, $\mathsf{U}_2$ or $\mathsf{U}_4$ is applicable, the algorithm terminates. 

\begin{algorithm} 
	\caption{$\EFX$ allocation with $\leq \floor{n/2}-1$ discarded goods} \label{alg1}
	Input : instance $(N, M, (\util_1, \ldots , \util_n))$ 
	
	Output: allocation $X$
	
	\begin{algorithmic}[1]
		\State $\allocs \leftarrow \langle\emptyset, \emptyset, \ldots, \emptyset\rangle$
		\State $R \leftarrow \langle \{1\}, \{2\}, \ldots, \{n\} \rangle$
		\While {$\mathsf{U}_0$ or $\mathsf{U}_1$ or $\mathsf{U}_2$ or $\mathsf{U}_4$ is applicable}
		\State Let $i$ be the minimum index s.t. $\mathsf{U}_i$ is applicable
		\State Update $(X, R)$ via Rule $\mathsf{U}_i$ 
		\EndWhile
		\State Return $X$
	\end{algorithmic}	
\end{algorithm}

In the rest of this section, we  describe the new updating rule, namely $\mathsf{U}_4$ and its properties. Next, based on the properties of these rules we prove Theorem \ref{main-theorem} which states that the final allocation is $\EFX$ and discards less than $\floor{n/2}$ goods.

\subsubsection*{Rule $\mathsf{U}_4$} 
In this rule, we consider the whole set of the goods in the pool to update the input $\sigma = (X, R)$. This rule states that if the set of the goods in the pool has a champion $i \in R_\ell$, we update the allocation as follows: we find a path from $r_\ell$ to $i$ and update the allocation through this path. Next, we allocate a subset of the goods in the pool to $i$ and return all the goods that belonged to $r_\ell$ before the update to the pool. Then we apply envy-elimination to obtain $\sigma' = (X', R')$. In the following, you can find a schematic view of Rule $\mathsf{U}_4$.

\vspace{0.5cm}
\begin{tcolorbox}[colback=blue!5,coltitle=blue!50!black,colframe=blue!25,
	title= \textbf{Rule $\mathsf{U}_4$}]
	
	{Preconditions:}
	\begin{itemize}
		\item $\sigma = (X, R)$ satisfies Properties \ref{prop1}, \ref{def-envy-conf}  and \ref{good-alloc}.
		\item There exists an agent who envies $P_X$.
	\end{itemize}
	\hrulefill
	
	{Process:}
	Let  $i \in R_\ell$ be a champion of $P$ and let $r_\ell = j_1\rightarrow j_2 \dots\rightarrow j_p = i$ be a path in $H_\sigma$. 
	\begin{itemize}
		\item For all $1 \leq t \leq p-1$, set  $\widetilde{X}_{j_t} = \alloc_{j_{t+1}}$.
		\item Set $\widetilde{X}_i = \champset{P}{i}$. 
		\item For all $k \notin \{j_1, j_2, \ldots, j_p\}$, set $\widetilde{X}_k = X_k$.
		\item Apply envy-elimination on $\widetilde{X}$.
	\end{itemize}
	\hrulefill
	
	Guarantees:
	\begin{itemize}
		\item $\sigma' = (X', R')$ satisfies Properties \ref{prop1}, \ref{def-envy-conf}, and \ref{good-alloc}.
		\item $\Phi(\sigma') {\lexlarger} \Phi(\sigma)$.
	\end{itemize}
	
\end{tcolorbox}
\vspace{0.5cm}

Since we apply envy-elimination in the last step of Rule $\mathsf{U}_4$, Lemma \ref{extended-props} implies that the resulting configuration satisfies Properties \ref{prop1}, \ref{def-envy-conf}, and \ref{good-alloc}.

\begin{corollary}[of Lemma \ref{extended-props}]\label{good-alloc-1}
	Let $\sigma' = (X', R')$ be the result of updating $\sigma = (X, R)$ via Rule $\textsf{U}_4$. Then, $\sigma' = (X', R')$ satisfies Properties \ref{prop1}, \ref{def-envy-conf}, and \ref{good-alloc}.
\end{corollary}

\begin{lemma}\label{phi-imp-1}
	Rule $\textsf{U}_4$ is $\Phi$-improving.
\end{lemma}
\begin{proof}
	We prove that the conditions of Lemma \ref{main-Phi-improving} hold.
	We have
	\begin{align*}
		\hbox{$\forall j \in R_1 \cup \ldots \cup R_{\ell-1}$, \quad $\widetilde{X}_j = X_j$.}
	\end{align*}
	Therefore, Condition \ref{cond1} of Lemma \ref{main-Phi-improving} holds.
	
	Without loss of generality, we can assume for all $k > \ell$, $v_{r_k}(X_{r_k}) > v_{r_\ell}(X_{r_\ell})$.\footnote{Otherwise, we can simply reorder the groups in $R$ which their representative agents have value $v_{r_\ell}(X_{r_\ell})$.} 
	For $k > \ell$ and $j \in R_k$ we have
	\begin{align*}
		v_j(\widetilde{X}_j) &= v_j(X_j) &\hbox{$\widetilde{X}_j = X_j$}\\ 
		&\geq v_{r_k}(X_{r_k}) &\hbox{$j \in R_k$}\\
		&> v_{r_\ell}(X_{r_\ell}).
	\end{align*} 
	Also if $j \in R_\ell \setminus \{j_1, j_2, \ldots, j_{t-1}\}$, we have
	\begin{align*}
	v_j(\widetilde{X}_j) &= v_j(X_j) &\hbox{$\widetilde{X}_j = X_j$}\\ 
	&> v_{r_\ell}(X_{r_\ell}).  &\hbox{$j \in R_\ell$}
	\end{align*} 
	Now assume $j=j_t$ for some $1 \leq t \leq p-1$. We have
	\begin{align*}
		v_{j_t}(\widetilde{X}_{j_t}) &= v_{j_t}(X_{j_{t+1}}) &\hbox{$\widetilde{X}_j = X_{j_{t+1}}$}\\
		&> v_{r_\ell}(X_{r_\ell}). &\hbox{Regular edge from $j_t$ to $j_{t+1}$}
	\end{align*}
	Finally if $j = i$, we have
	\begin{align*}
		v_i(\widetilde{X}_i) &= v_i(\champset{P}{i}) &\hbox{$\widetilde{X}_i = \champset{P}{i}$}\\
		&> v_{r_\ell}(X_{r_\ell}). &\hbox{Definition \ref{e-champion}}
	\end{align*}
	Therefore,
	\begin{align*}
		\hbox{$\forall j \in R_\ell \cup \ldots \cup R_{|R|}$, \quad $v_j(\widetilde{X}_j) > v_{r_\ell}(X_{r_\ell})$.}
	\end{align*}
	Thus, Condition \ref{cond2} of Lemma \ref{main-Phi-improving} holds too and hence, $\Phi(\sigma') \lexlarger \Phi(\sigma)$.
\end{proof}

\subsection{Bounding the Number of Discarded Goods}
As we mentioned, when none of the updating rules is applicable, we terminate the algorithm. As we proved in Lemma \ref{good-alloc-0} and Corollaries \ref{good-alloc-2}, \ref{good-alloc-3}, and \ref{good-alloc-1}, the allocation after each update is $\EFX$. Therefore, the final allocation preserves the $\EFX$ property as well. Also, since by Lemmas \ref{phi-imp-0}, \ref{phi-imp-2}, \ref{phi-imp-3}, and \ref{phi-imp-1} the value of $\Phi(\sigma)$ increases after each update, Algorithm \ref{alg1} terminates after a finite number of updates. It remains to prove that the number of  remaining goods in the pool is less than $\floor{n/2}$ and no agent envies the pool. We prove this in Lemma \ref{upper_bound}.

\begin{lemma}\label{upper_bound}
	Let $P$ be the pool of unallocated goods at the end of Algorithm \ref{alg1}. Then, we have $|P| < \floor{n/2}$. Also, no agent envies $P$.
\end{lemma}
\begin{proof}
		Since Rule $\textsf{U}_3$ is not applicable, for every agent $i \in R_\ell$ we have $\valu_i(\alloc_i) \geq v_{r_\ell}(X_{r_\ell}) \geq \valu_i(P)$.
	Now, we show that $|P| <\floor{n/2}$ also holds.  
	Towards a contradiction, assume that $|P| \geq \floor{n/2}$.
	If $|R| \leq \floor{n/2}$, by Lemma \ref{small-pool} at least one of Rules $\mathsf{U}_0$ or $\mathsf{U}_1$ is applicable which contradicts the termination of the algorithm. Therefore, we have $|R| > \floor{n/2}$.
	Let $S \subseteq \{1, 2, \ldots, |R|\}$ be the set with the following properties:
	\begin{itemize}
		\item for all agents $i$ with $v_i(P) > 0$, if $i \in R_\ell$ then $\ell \in S$, and
		\item $|S|$ is minimum.
	\end{itemize}
	Note that $R_i \cap R_j = \emptyset$ for all different $i$ and $j$ and $\cup_{\ell \in S} |R_{\ell}| \leq n$.
	Hence, by the Pigeonhole principle, if $|R| > \floor{n/2}$, there is an index $i \in S$ such that $|R_i| = 1$, that is, $R_i = \{r_i\}$. Since $|S|$ is minimum, we have $v_{r_i}(P) > 0$; otherwise we could remove $i$ from $S$ which contradicts the minimality of $S$. 
	But this means that $\mathsf{U}_4$ is applicable, which contradicts the termination assumption.
\end{proof}

Finally, Lemma \ref{upper_bound} together with the fact that the final allocation is $\EFX$ yields Theorem \ref{main-theorem}.
\begin{restatable}{theorem}{efxnt}\label{main-theorem}
	Assuming that the valuations are restricted additive, Algorithm \ref{alg1} returns an allocation $X$ and a pool $P$ of unallocated goods such that 		
	\begin{itemize}
		\item $\allocs$ is $\EFX$, and
		\item $\valu_i(\alloc_i) \geq \valu_i(P)$ for every agent $i$, and
		\item $|P| < \floor{n/2}$.
	\end{itemize}
\end{restatable}

\section{Optimality of $\EFX$ Allocations} \label{OPTvsEFX}
In Section \ref{efxn2}, we presented an algorithm that finds a partial $\EFX$ allocation. When it comes to efficiency, it is not difficult to see that at the end of this algorithm only representative agents can own wasted goods and all other agents have no wasted good in their bundle. A natural question that arises is whether there exists $\EFkX$ allocations that are more efficient.
In this section, we address this question: can we find an allocation that is Pareto efficient and satisfies $\EFkX$?  

Our answer to this question is negative. 
We show that for some instances guaranteeing both Pareto efficiency and $\EFkX$ is not possible. Indeed, our result is stronger; we show that for any constant $c<1$ and any integer $k$, there are instances in which guaranteeing Pareto efficiency and a $c$-approximation of $\EFkX$ is not possible. We start by defining $c$-$\mathsf{EFkX}$ allocations. 

\begin{definition}\label{cefkxdef}
	An allocation $X = \langle \alloc_1, ... , \alloc_n \rangle$ is $c$-$\mathsf{EFkX}$, if for all $i \neq j$ and every collection of $\ell = min(k, |X_j|)$ distinct goods $g_1,g_2,\ldots,g_\ell \in \alloc_j$  we have 
	$
	v_i(\alloc_i) \geq c v_i(\alloc_j \setminus\{g_1,g_2,\ldots,g_\ell\}).
	$
\end{definition}

Roughgarden and Plaut \cite{plaut2020almost} gave an example in which no $\EFX$ and Pareto efficient allocation exists.
We generalize this example and for any integer $k$ and any constant $c \leq 1$, we present an instance which admits no $c$-$\mathsf{EFkX}$ and Pareto efficient allocation.	

\begin{example}
	Let $\agents = \{1, 2\}$ and $\goods = \{\good, \good_1, \dots , \good_k, \good'_1, \dots , \good'_k\}$. Also let $\util_1(\good) = \util_2(\good) = (\frac{1}{c}+1)^k$, $\util_1(\good_i) = \util_2(\good'_i) = (\frac{1}{c}+1)^{i-1}$ for all $i \in [k]$ and  $\util_1(\good'_i) = \util_2(\good_i) = 0$ (see Table \ref{efkxcounter}). In any Pareto efficient allocation, $\good_i$ must be allocated to agent $1$ and $\good'_i$ must be allocated to agent $2$ for all $i \in [k]$. Without loss of generality (by symmetry), assume agent $1$ receives $\good$. We have,
	\begin{align*}
	\util_2(\alloc_1) &= \util_2(\good) \\ 
	&= \util_2(\alloc_1 \backslash \{\good_1, \good_2, \dots , \good_k\}\}) \\ 
	&= (\frac{1}{c}+1)^k .
	\intertext{Also,}
	\util_2(\alloc_2) &= 1 + (\frac{1}{c}+1) + \dots + (\frac{1}{c}+1)^{k-1} \\
	&= c((\frac{1}{c}+1)^k-1). \\
	\intertext{Therefore,}
	c \util_2(\alloc_1 \backslash \{\good_1, \good_2, \dots , \good_k\}) &> \util_2(\alloc_2).
	\end{align*}
	Thus, the allocation is not $c$-$\mathsf{EFkX}$.
	\begin{table}
		\centering
		\begin{tabular}{ c|c c c c c c c c c c}
			&$\good$ & $\good_1$ & $\good_2$ & $\dots$ & $\good_k$ & $\good'_1$ & $\good'_2$ & $\dots$ & $\good'_k$ &\\
			\hline
			$\util_1$& $(\frac{1}{c}+1)^k$ & $1$ & $\frac{1}{c}+1$ & $\dots$ & $(\frac{1}{c}+1)^{k-1}$ & $0$ & $0$ & $\dots$ & $0$ &\\
			$\util_2$& $(\frac{1}{c}+1)^k$ & $0$ & $0$  & $\dots$ & $0$ & $1$ & $\frac{1}{c}+1$ & $\dots$ & $(\frac{1}{c}+1)^{k-1}$ &\\
		\end{tabular}
		\caption{Counter example for $c$-$\mathsf{EFkX}$ and Pareto efficiency}
		\label{efkxcounter}
	\end{table}
\end{example}

In the last part of this section, we define $\EFXp$ (which is in fact the original definition of $\EFX$ introduced by Caragiannis \etal \cite{caragiannis2019unreasonable}) which is a relaxation of $\EFX$ and present an algorithm that guarantees $\EFXp$ and Pareto efficiency at the same time when the valuations are restricted additive. 

\begin{definition}\label{efxpdef}
	An allocation $X = \langle \alloc_1, ... , \alloc_n\rangle$ is $\EFXp$ if for all agents $i$ and $j$ and all goods $\good \in \alloc_j$ such that $\valu_i(\alloc_j) > \valu_i(\alloc_j \setminus \{\good\})$, $\valu_i(\alloc_i) \geq \valu_i(\alloc_j \setminus \{\good\})$.
\end{definition}

Algorithm \ref{efxalg} shows a pseudocode of our algorithm for finding an $\EFXp$ and Pareto efficient allocation. 
The algorithm is quite simple. We start with an empty allocation and allocate the goods in non-increasing order of their values. For each good $\good$, we allocate it to agent $i$ such that $\util_i(\alloc_i)$ is minimum and $\util_i(\good) = \valu(\good)$. It is easy to see that the algorithm runs in polynomial time. The sorting of $m$ goods can be done in $\mathcal{O}(m \log m)$ and then for each good $g$ we need to find $\argmin_{j: \valu_j(\good_\ell) = \valu(\good_\ell)}(\valu_j(\alloc_j))$ which can be done in $\mathcal{O}(n)$. Thus in total the running time of Algorithm \ref{efxalg} is $\mathcal{O}(m \log m + mn)$.

In Definition \ref{rest-add-def}, we have $v(g)>0$ for all goods $g$. Note that if for a good $g$ we have that $v_i(g) = 0$ for all agents $i$, we can assume any value for $v(g)$ since it is irrelevant. In order to prevent some unnecessary case distinctions, in this section we assume $v(g) = 0$ for all goods $g$ such that for all agents $i$, $v_i(g)=0$.

\begin{algorithm}
	\caption{Complete $\EFXp$ allocation} \label{efxalg}
	Input : instance $(N, M, (\util_1, \ldots , \util_n))$ 
	
	Output: allocation $\allocs = \langle\alloc_1, ... , \alloc_n\rangle$
	
	\begin{algorithmic}[1]
		\State Let $\allocs = \langle \emptyset, \emptyset, ... , \emptyset \rangle$
		\State Let $\valu(\good_1) \geq \valu(\good_2) \geq \ldots \geq \valu(\good_m)$
		\For {$\ell \leftarrow 1$ to $m$}
		\State Let $i = \argmin_{j: \valu_j(\good_\ell) = \valu(\good_\ell)}(\valu_j(\alloc_j))$  	
		\State $\alloc_i \leftarrow \alloc_i \cup \{\good_\ell\}$
		\EndFor
	\end{algorithmic}
\end{algorithm}

\begin{observation}\label{non-wasteful}
	Let $\alloc$ be the allocation returned by Algorithm \ref{efxalg}. Then $\alloc$ is non-wasteful and therefore, Pareto efficient.
\end{observation}

Finally, we show that the allocation returned by Algorithm \ref{efxalg} is $\EFXp$. The proof is by induction. We start with the empty allocation which is $\EFXp$. Let us assume allocation $\alloc$ is $\EFXp$ before allocating $\good_\ell$ to agent $i$ (induction hypothesis). It suffices to prove for all agents $j$ and $h \in \alloc_i \cup \{\good_\ell\}$ such that $\valu_j(h) = \valu(h)$,
$$\valu_j(\alloc_j) \geq \valu_j((\alloc_i \cup \{\good_\ell\}) \setminus \{h\}).$$

If $\valu_j(\good_\ell) = 0$, we have
\begin{align*}
\valu_j(\alloc_j) &\geq \valu_j(\alloc_i \setminus \{h\}) &\mbox{induction hypothesis}\\
&= \valu_j(\alloc_i \cup \{\good_\ell\} \setminus \{h\}). &\mbox{$\valu_j(\good_\ell) = 0$}
\end{align*}

Otherwise, we have
\begin{align*}
\valu_j(\alloc_j) &\geq \valu_i(\alloc_i) &\mbox{by the choice of $i$}\\
&= \valu(\alloc_i) &\mbox{Observation \ref{non-wasteful}} \\
&\geq \valu(\alloc_i) + \valu(\good_\ell) - \valu(h) &\mbox{$h=\good_{\ell'}$ for $\ell' \leq \ell$}\\
&= \valu((\alloc_i \cup \{\good_\ell\}) \setminus \{h\}) \\
&\geq \valu_j((\alloc_i \cup \{\good_\ell\}) \setminus \{h\}).
\end{align*}

This concludes that the output of Algorithm \ref{non-wasteful} is both $\EFXp$ and Pareto efficient.\footnote{We have been notified (personal communication, July 2022) that the authors of \cite{beyond} have this result in paralel.}

\begin{theorem}
	Assuming that the valuations are restricted additive, there exists a polynomial-time algorithm that returns an $\EFXp$ and Pareto efficient allocation. 
\end{theorem}

\newpage
\bibliographystyle{plain}
\bibliography{nsw}
\clearpage
\appendix

\section{Useful Lemmas on Comparing  Vectors Lexicographically}
\begin{lemma}\label{lex1}
	Let $\sigma = (X, R)$ and $\sigma' = (X', R')$ be two envy-compatible configurations and let $\{r_1, r_2, \ldots, r_{|R|}\}$ and $\{r'_1, r'_2, \ldots, r'_{|R'|}\}$ be the sets of their representative agents respectively. If $\{r'_1, \ldots, r'_{|R'|}\} \subsetneq \{r_1, \ldots, r_{|R|}\}$ and for all $i \in \{r'_1, \ldots, r'_{|R'|}\}$, $v_i(X'_i) \geq v_i(X_i)$, $\Phi(\sigma') \lexlarger \Phi(\sigma)$.
\end{lemma}
\begin{proof}
	Note that $|R| > |R'|$. If $v_{r_\ell}(X_{r_\ell}) = v_{r'_\ell}(X'_{r'_\ell})$ for all $1 \leq \ell \leq |R'|$, then since $v_{r_{|R'|+1}}(X_{r_{|R'|+1}}) < +\infty$, we have $\Phi(\sigma') \lexlarger \Phi(\sigma)$.
	
	Otherwise, let $\ell$ be the smallest index such that 
	\begin{align*}
	v_{r_\ell}(X_{r_\ell}) \neq v_{r'_\ell}(X'_{r'_\ell}). 
	\end{align*}
	If $v_{r'_\ell}(X'_{r'_\ell}) > v_{r_\ell}(X_{r_\ell})$, then $\Phi(\sigma') \lexlarger \Phi(\sigma)$. 
	
	Now assume otherwise. We have
	\begin{align}
		v_{r_\ell}(X_{r_\ell}) > v_{r'_\ell}(X'_{r'_\ell}). \label{ineq}
	\end{align} 
	For all $1 \leq k \leq \ell$, we have
	\begin{align*}
		v_{r'_k}(X_{r'_k}) &\leq v_{r'_k}(X'_{r'_k}) &\mbox{Condition of Lemma \ref{lex1}} \\ 
		&\leq v_{r'_\ell}(X'_{r'_\ell}) &\mbox{$k \leq \ell$} \\
		&< v_{r_\ell}(X_{r_\ell}) &\mbox{Inequality \eqref{ineq}} .		
	\end{align*}
	Therefore, we have $r'_k \in \{r_1, \ldots, r_{\ell-1}\}$ for all $1 \leq k \leq \ell$ which is a contradiction since $|\{r_1, \ldots, r_{\ell-1}\}| = \ell-1$. 
\end{proof}

\begin{lemma}\label{lex3}
	Let $\sigma = (X, R)$ and $\sigma' = (X', R')$ be two envy-compatible configurations and let $\{r_1, r_2, \ldots, r_{|R|}\}$ and $\{r'_1, r'_2, \ldots, r'_{|R'|}\}$ be the sets of their representative agents respectively. Assume there exists $1 \leq \ell \leq |R|$ such that the following properties hold:
	\begin{enumerate}
		\item for all $i \in R_1 \cup \ldots \cup R_{\ell-1}$, if $i \notin \{r_1, r_2, \ldots, r_{|R|}\}$ then $i \notin \{r'_1, r'_2, \ldots, r'_{|R'|}\}$, and
		\item for all $i \in R_1 \cup \ldots \cup R_{\ell-1}$, if $i \in \{r'_1, r'_2, \ldots, r'_{|R'|}\}$, $v_i(X'_i) = v_i(X_i)$, and
		\item for all $i \in R_\ell \cup \ldots \cup R_k$, if $i \in \{r'_1, r'_2, \ldots, r'_{|R'|}\}$, $v_i(X'_i) > v_{r_\ell}(X_{r_\ell})$.
	\end{enumerate}
	Then, $\Phi(\sigma') \lexlarger \Phi(\sigma)$.
\end{lemma}
\begin{proof}
	Let $k$ be the largest index such that $v_{r'_k}(X'_{r'_k}) \leq v_{r_\ell}(X_{r_\ell})$. Property $3$ implies that $\{r'_1, \ldots, r'_k\} \subseteq R_1 \cup \ldots \cup R_{\ell-1}$ and Property $1$ implies that $\{r'_1, \ldots, r'_k\} \subseteq \{r_1, \ldots, r_{\ell-1}\}$.
	
	\begin{itemize}
		\item \textbf{Case 1: $\{r'_1, \ldots, r'_k\} \subsetneq \{r_1, \ldots, r_{\ell-1}\}$}
		
		Note that $\ell-1 > k$. If $v_{r_i}(X_{r_i}) = v_{r'_i}(X'_{r'_i})$ for all $1 \leq i \leq k$ and $|R'| > k$, we have
		\begin{align*}
		v_{r'_{k+1}}(X_{r'_{k+1}}) &> v_{r_\ell}(X_{r_\ell}) &\mbox{choice of $k$} \\
		&\geq v_{r_{k+1}}(X_{r_{k+1}}). &\mbox{$\ell > k+1$}
		\end{align*}
		Thus, $\Phi(\sigma') \lexlarger \Phi(\sigma)$. Also if $|R'| = k$, since $+\infty > v_{r_{k+1}}(X_{r_{k+1}})$, $\Phi(\sigma') \lexlarger \Phi(\sigma)$.
		
		Now let $i \leq k$ be the smallest index such that 
		\begin{align*}
		v_{r_i}(X_{r_i}) \neq v_{r'_i}(X'_{r'_i}). 
		\end{align*}
		If $v_{r'_i}(X'_{r'_i}) > v_{r_i}(X_{r_i})$, then $$(v_{r'_1}(X'_{r'_1}), \ldots , v_{r'_k}(X'_{r_k})) \lexlarger (v_{r_1}(X_{r_1}), \ldots ,v_{r_{\ell-1}}(X_{r_{\ell-1}})),$$
		and hence $\Phi(\sigma') \lexlarger \Phi(\sigma)$.
		
		So assume we have
		\begin{align}
		v_{r_i}(X_{r_i}) > v_{r'_i}(X'_{r'_i}). \label{ineq-prim}
		\end{align} 
		For all $1 \leq j \leq i$, we have
		\begin{align*}
		v_{r'_j}(X_{r'_j}) &= v_{r'_j}(X'_{r'_j}) &\mbox{Condition of Lemma \ref{lex3}} \\ 
		&\leq v_{r'_i}(X'_{r'_i}) &\mbox{$j \leq i$} \\
		&< v_{r_i}(X_{r_i}) &\mbox{Inequality \eqref{ineq-prim}} .		
		\end{align*}
		Therefore, we have $r'_j \in \{r_1, \ldots, r_{i-1}\}$ for all $1 \leq j \leq i$ which is a contradiction since $|\{r_1, \ldots, r_{i-1}\}| = i-1$. 
			
		\item \textbf{Case 2: $\{r'_1, \ldots, r'_k\} = \{r_1, \ldots, r_{\ell-1}\}$}
		
		Note that by Property 2, for all $i \in \{r'_1, \ldots, r'_k\}$, $v_i(X'_i)=v_i(X_i)$. 
		If $|R'| = k$, since $+\infty > v_{r_\ell}(X_{r_\ell})$, $\Phi(\sigma') \lexlarger \Phi(\sigma)$.
		If $|R'| > k$, by the choice of $k$ we have $v_{r'_{k+1}}(X'_{r'_{k+1}}) > v_{r_\ell}(X_{r_\ell})$. Thus, $\Phi(\sigma') \lexlarger \Phi(\sigma)$. 
	\end{itemize}
\end{proof}

\begin{lemma} \label{main-Phi-improving}
	Let $\sigma = (X, R)$ be an admissible configuration. Assume we modify $X$ to $\widetilde{X}$ such that there exists $1 \leq \ell \leq |R|$ for which the following properties hold:
	\begin{enumerate}
		\item \label{cond1} for all $i \in R_1 \cup \ldots \cup R_{\ell-1}$, $\widetilde{X}_i = X_i$, and
		\item \label{cond2} for all $i \in R_\ell \cup \ldots \cup R_{|R|}$, $v_i(\widetilde{X}_i) > v_{r_\ell}(X_{r_\ell})$.
	\end{enumerate}
	Let $\sigma' = (X', R')$ be the result of applying envy-elimination on $\widetilde{X}$. Then $\Phi(\sigma') \lexlarger \Phi(\sigma)$.
\end{lemma}
\begin{proof}
	We prove that the properties of Lemma \ref{lex3} hold. 
	First we prove that
	\begin{align*}
		\text{for all $i \in R_1 \cup \ldots \cup R_{\ell-1}$, \quad if $i \notin \{r_1, r_2, \ldots, r_{|R|}\}$ then $i \notin \{r'_1, r'_2, \ldots, r'_{|R'|}\}$.}
	\end{align*}
	We use the following claim.
	\begin{claim}\label{subset-sources}
		For every $1 \leq k < \ell$ and every agent $i \in R_k$, 
		\begin{itemize}
			\item if $i \in R'_t$, then $v_{r'_t}(X'_{r'_t}) \leq v_{r_k}(X_{r_k})$, and
			\item if $i \notin \{r_1, r_2, \ldots, r_{|R|}\}$, then $i \notin \{r'_1, r'_2, \ldots, r'_{|R'|}\}$.
		\end{itemize}
	\end{claim} 
	The proof of the claim is by induction on the distance of $i$ from $r_k$. If $i = r_k$, then either $i \in \{r'_1, \ldots, r'_{|R'|}\}$ and $v_{r'_t}(X'_{r'_t}) = v_{r_k}(X_{r_k})$ or, $v_{r'_t}(X'_{r'_t}) < v_{r_k}(X'_{r_k}) \leq v_{r_k}(X_{r_k})$. Now assume $i \neq r_k$, and $j$ is the preceding node of $i$ in a shortest path from $r_k$ to $i$ in $H_\sigma$. By induction assumption, $j \in R'_t$ for $v_{r'_t}(X'_{r'_t}) \leq v_{r_k}(X_{r_k})$. 
	The bundle of $i$ is not changed before the envy-elimination. Also since $X$ is admissible and $i$ is not a representative, $X_i$ has no wasted good. Therefore, the bundle of $i$ is not changed even after the removal of wasted goods in the envy-elimination process. We have $v_i(X_i) > v_{r_k}(X_{r_k}) \geq v_{r'_t}(X'_{r'_t})$. Therefore, either $i$ is already added to some $R'_{t'}$ for $t' \leq t$ or we set $X'_i := \champset{X_i}{j}$ and $i$ gets added to $R'_t$. Thus, the claim holds.
	
	For the second property, note that when running envy-elimination, the bundles of the agents in $\{r'_1, r'_2, \ldots, r'_{|R'|}\}$ do not change except for the removal of wasted goods. Therefore, 
	\begin{align*}
		\text{for all $i \in R_1 \cup \ldots \cup R_{\ell-1}$, \quad if $i \in \{r'_1, r'_2, \ldots, r'_{|R'|}\}$, $v_i(X'_i) = v_i(\widetilde{X}_i) = v_i(X_i)$.}
	\end{align*}
	We know that for all $j \in \{r'_1, r'_2, \ldots, r'_{|R'|}\} \cap (R_\ell \cup \ldots \cup R_k)$, 
	\begin{align}
	v_j(\widetilde{X}_j) > v_{r_\ell}(X_{r_\ell}). \label{Yineq}
	\end{align}	
	Hence, for all $j \in \{r'_1, r'_2, \ldots, r'_{|R'|}\} \cap (R_\ell \cup \ldots \cup R_k)$, 
	\begin{align*}
		v_j(X'_j) &= v_j(\widetilde{X}_j) \\
		&> v_{i_\ell}(X_{i_\ell}). &\hbox{Inequality \eqref{Yineq}}
	\end{align*}
	Therefore, by Lemma \ref{lex3}, $\Phi(\sigma') \lexlarger \Phi(\sigma)$.
\end{proof}

\end{document}